\documentclass[12pt]{article}
\usepackage{amsfonts}
\usepackage{amssymb}
\usepackage{graphicx}
\usepackage{amsmath}
\usepackage{amstext}
\usepackage{geometry}
\usepackage[onehalfspacing]{setspace}
\usepackage[labelfont={sc}]{caption}
\usepackage{hyperref}
\usepackage{natbib}

\newtheorem{corollary}{Corollary}

\newtheorem{definition}{Definition}

\newtheorem{lemma}{Lemma}

\newtheorem{proposition}{Proposition}

\newenvironment{proof}[1][Proof]{\textbf{#1.} }{\ \rule{0.5em}{0.5em}}

\newcommand{\ta}{\theta}
\newcommand{\la}{\lambda}
\begin{document}

\title{Data, Competition, and\ Digital Platforms\thanks{%
We acknowledge financial support through NSF\ Grant SES-1948692, the Omidyar
Network, and the Sloan Foundation. We thank audiences at Bank of Canada,
Berkeley, Indiana, Northwestern, NYU, Paris-Panth\'{e}on, Toulouse, as well
as participants in the CEPR\ Virtual IO\ Seminar, the Luohan Academy
Webinar, the FTC Microeconomics Conference, the Asia-Pacific Industrial
Organization Conference, the 2022 Market Innovation Workshop, and the 2022
NBER\ Economics of Privacy Conference. We are grateful to our discussants
Andrei Hagiu, Alessandro Pavan, and Vasiliki Skreta for helpful comments and
suggestions, and to Lucas Barros, Andrew Koh, and Nick Wu for stellar
research assistance.}}
\author{Dirk Bergemann\thanks{%
Department of Economics, Yale University, New Haven, CT 06511,
dirk.bergemann@yale.edu} \and Alessandro Bonatti\thanks{%
Sloan School of Management, MIT, Cambridge, MA 02142, bonatti@mit.edu}}
\date{\today }
\maketitle

\begin{abstract}
We analyze digital markets where a monopolist platform uses data to match
multiproduct sellers with heterogeneous consumers who can purchase both on
and off the platform. The platform sells targeted ads to sellers that
recommend their products to consumers and reveals information to consumers
about their values. The revenue-optimal mechanism is a managed advertising
campaign that matches products and preferences efficiently. In equilibrium,
sellers offer higher qualities at lower unit prices on than off the
platform. Privacy-respecting data-governance rules such as organic search
results or federated learning can lead to welfare gains for
consumers.\bigskip


\bigskip \bigskip
\end{abstract}

\newpage

\section{Introduction}

\subsection{Motivation}

The role of data in shaping competition in online markets has become a
critical issue in both economics and policy. Digital platforms such as
Amazon, Facebook, and Google in the United States, and Alibaba, JD, and
Tencent in China have collected increasingly large and precise datasets.
These platforms operate as matching engines that connect viewers and
sellers. They monetize their data by selling sponsored content and targeted
advertising. The quality of a platform's data allows better pairing of
viewers and sellers. Digital platforms not only serve as gatekeepers of
information online but also act as competition managers.\footnote{%
The European Union's Digital Markets Act establishes a set of narrowly
defined objective criteria for qualifying a large online platform as a
so-called \textquotedblleft gatekeeper\textquotedblright . These criteria
will be met if a company: (a) has a strong economic position in the EU\
market, (b) has a strong intermediation position, (c) has (or is about to
have) an entrenched and durable position in the market.}

Regulators fear that platforms may leverage their gatekeeper position to
increase merchants' market power, thereby raising their willingness to pay
for advertising.\footnote{%
The report by \cite{crem19} explicitly warns that \textquotedblleft one
cannot exclude the possibility that a dominant platform could have
incentives to sell \textquotedblleft monopoly positions\textquotedblright\
to sellers by showing buyers alternatives which do not meet their
needs.\textquotedblright} The optimal regulatory response to the current
business practices of digital platforms, if any, depends on the answers to a
number of open questions, including the following: how does the precision of
a digital platform's data affect the creation and distribution of surplus,
both on and off the platform? How do these effects depend on the intensity
of competition among sellers? How do they depend on the mechanisms for
collecting and sharing consumer data?

In this paper, we develop a model of an intermediated online marketplace and
trace how a data-rich platform creates and distributes surplus among market
participants. Our goal is to provide a tractable and flexible framework to
study digital markets where different privacy regimes can be compared. Our
model captures three ubiquitous features of digital platforms. First, nearly
every platform leverages the informational advantage to personalize the 
\emph{sponsored content} at the individual consumer level through managed
advertising campaigns.\footnote{%
Sponsored links are the only source of revenue for pure advertising
platforms, including display advertising networks such as Google, Meta,
Microsoft, Twitter, Tiktok, YouTube, and Criteo. Sponsored content is also a
significant revenue generator for several retail platforms that charge
merchant fees (eBay, Wayfair, Booking, Orbitz, Amazon) and the only source
of revenue for Alibaba's Taobao shopping platform.} Second, while price
discrimination is rare, targeted advertising and personalized
recommendations amount to \emph{product steering}.\footnote{\cite{dokm22}
document the effect of personalized recommendations on a retail platform,
and \cite{raval20} illustrates a recent shift in eBay policy.} Third, most
sellers have parallel sales channels, i.e., consumers can buy their products
both on and off digital platforms.

To capture these features, we consider a market with differentiated
quality-pricing sellers. Consumers have heterogeneous preferences for the
sellers' product lines but are imperfectly informed about their own values.
The platform's data identify the most valuable consumer-seller pair and the
most valuable product within that seller's product line. However, these data
also create the potential for product steering, whereby consumers who are
perceived to be of high value receive offers to buy higher-quality and
higher-priced goods.

A key innovation in our model is that the platform actively manages the
sellers' advertising campaigns. Managed campaigns are emerging as the
predominant mode of selling advertisements in real-world digital markets:
sellers set a fixed advertising budget, specify high-level objectives for
their campaigns, and leave the task of bidding to \textquotedblleft
auto-bidders\textquotedblright\ offered by the platform.\footnote{%
Most of online advertising is traded programatically, often through a
mechanism that maximizes the total number of clicks or conversions subject
to a return on ad spend (ROAS) constraint (\href{https://support.google.com/google-ads/answer/6268637}%
{https://support.google.com/google-ads/answer/6268637}). For recent work on
programmatic or algorithmic bidding, see \cite{agga19}, \cite{bals21}, and 
\cite{deng21}.} In our model, the digital platform receives an advertising
budget from each seller. The advertising budget then generates sponsored
listings for each specific product of that seller. In particular, the
platform advertises to each consumer the product generating the highest
social surplus for that consumer. This mechanism is akin to the managed
advertising campaigns prevalent in sponsored search and other forms of
digital advertising. In turn, these are frequently implemented through an
automated bidding mechanism by platforms such as Google, Facebook, or Tiktok.

Each seller also has a pool of consumers who shop off the platform and face
search costs. In the tradition of the \cite{diamond71} model, these
consumers have zero cost to search the first seller and positive cost to
visit additional sellers. In equilibrium, each consumer only visits the
website of a single seller---the one whose products they value most. The
presence of the off-platform sales channel restrains sellers' ability to
extract consumer surplus on the platform because the on-platform consumers
can seamlessly move from the platform to individual websites off the
platform. In particular, the more the seller wants to trade with the loyal
consumers off the platform, the less flexibility it has to offer targeted
promotions (and prices) on the platform. Thus, the consumer's choice of
sales channel limits the scope for price discrimination.

\subsection{Results}

The platform's informational advantage over consumers and the search
frictions on the platform, no matter how small, give the digital platform
significant bargaining power over sellers. Our first main result shows that
the platform can completely control consumers' shopping behavior and steer
them away from sellers who do not submit an advertising budget (Proposition %
\ref{prop_consider}). Consumers understand the managed-campaign mechanism
and expect that in equilibrium, advertised products generate the highest
value. In the presence of search costs, consumers only consider buying from
the advertised brand, whether on-platform or off-platform. As a result, the
platform restricts competition among sellers, as each seller only faces
consumers who are most interested in their products and competes with their
own off-platform offers only (Proposition \ref{menu}).

This leads to sellers facing an additional opportunity cost of generating
surplus off the platform. Not only must they concede information rents to
off-platform consumers, but they must also lower their prices on the
platform. This has two welfare consequences. First, the equilibrium quality
levels of off-platform products are distorted downward from the efficient
levels even more than in the model proposed by \cite{muro78}. Second, the
platform is able to extract most of the surplus it generates, as it only
needs to compensate sellers for the additional distortions in their
off-platform menus of products (Proposition \ref{optimalmech}).

Next, we examine the platform's dual gatekeeper role by considering two
sources of the platform's bargaining power: its information advantage over
consumers and sellers, and consumers' search costs off the platform. We
first show that it would be against the platform's interest to provide
consumers with information about non-sponsored products. We assume that
consumers on the platform observe their values perfectly, so the platform
cannot steer their behavior away from their favorite seller even if that
seller rejects the platform's offer. Complete information for on-platform
consumers does not change the equilibrium prices or products, but it reduces
the platform's fees (Proposition \ref{prop_knowntype}). We then assume that
the platform offers organic links that advertise all off-platform prices and
products to all on-platform consumers. We show that the provision of price
information introduces menu competition among sellers, reducing all
equilibrium prices, both on and off the platform, as well as the platform's
fees (Proposition \ref{prop_organic}).

We also investigate how data governance, the rules governing how consumer
data can be collected and deployed, influences the creation and distribution
of social surplus. In particular, we discuss the implications of
cohort-based advertising, as the outcome of federated learning, compared to
personalized advertising. We show that cohort-based advertising, which
protects privacy and allows consumers to retain an information advantage
over sellers on the platform, improves consumer surplus (Proposition \ref%
{floc}).

So far, the platform has been using all the additional information for
product steering and pricing recommendations. We then explore whether the
platform can do even better by employing the additional information only
partially and stochastically. When the on-platform market is large compared
to the off-platform market, using complete information indeed maximizes the
platform's advertising revenue (Proposition \ref{prop_fullinfo}). However,
when consumers only know the prior distribution and the off-platform market
is sufficiently large, the platform can increase its revenue by offering a
more limited disclosure policy, which we fully characterize (Proposition \ref%
{prop_id}).

Finally, we examine the effects of the platform size. We show that the
distortions in off-platform quality become more severe as the fraction of
on-platform consumers grows (Proposition \ref{compstat}). Holding prices
fixed, consumers benefit from shopping on a better-informed platform and
obtaining higher-quality matches. However, as more consumers join the
platform, prices rise both on and off the platform.

\subsection{Related Literature}

This paper is most closely related to the literature on information
gatekeepers pioneered by \cite{baye2001} and on the conflict of interest
between intermediaries and the consumers they serve. Many recent
contributions---including \cite{arzh11}, \cite{cosz22}, \cite{deta19}, \cite%
{gopa16}, \cite{gur22}, \cite{haju11}, \cite{inot12a}, \cite{inot12b}, \cite%
{tony22}, \cite{rase10}, and \cite{shi22}---analyze the steering role of
platforms that strategically modify search results, e.g., to match consumers
with the sellers that pay the largest commissions.\footnote{%
The trade-off between value creation through personalization and consumer
surplus extraction is also central in \cite{hive21} and \cite{ichi20aer}.}

The provision of information by a digital platform is central to the model
of \cite{dede16}, who examine a platform's incentives to provide match-value
information to differentiated sellers in a second-price auction model. More
recently, \cite{tewr22} study the signaling role of ranking the search
results, and \cite{zhmosh22} study the impact of a platform's privacy
policies on the downstream competition within and across product categories.
Relative to these papers, we allow the platform to provide information
directly to the consumer, e.g., through product reviews. Moreover, the
multiproduct sellers in our model can use the platform's information to
tailor their quality level to the consumer's preferences. This allows us to
capture surplus creation and product steering \emph{within} a match.
Finally, relative to the papers above, our model focuses on sponsored links
and advertising platforms. Hence, sellers pay fees (or bids) that do not
vary with the prices of their products.\footnote{%
In \cite{gopa22}, a monopolist platform elicits sellers' willingness to pay
for qualified consumer eyeballs through a nonlinear tariff for advertising
space, which is analogous to our managed campaign. In contrast, we
explicitly model consumer search and determine the equilibrium prices for
the consumers.}

A recent body of work including \cite{chjk19}, \cite{ammo21}, \cite{ichi21}, 
\cite{kiph21}, and \cite{bebg22} documented the \emph{data externalities}
that consumers impose on each other when they share their information with a
digital platform. In the present paper, the growth of a platform's database
(e.g., through the participation of other consumers) influences the ability
to match products to tastes but also affects each consumer's outside option.
We trace the implications of these new data externalities for product-line
design under alternative privacy regimes.

The forces at work in our paper are also related to a growing literature on
showrooming, product lines, and multiple sales channels. Prominent
contributions on these topics include \cite{bash20}, \cite{idem20}, \cite%
{mish19}, and \cite{wawr20}. In particular, \cite{anbe21} introduce the
self-preferencing problem by letting the platform choose whether to be
hybrid, i.e., to sell the private label products.\footnote{\cite{guti22}
estimates an equilibrium model of Amazon's hybrid retail platform.} Unlike
in these papers, the sellers in our model are concerned about showrooming
because the opportunity to sell on the platform benefits them through the
added value of making personalized offers.

Our analysis of parallel sales channels is also related to the papers on
\textquotedblleft partial mechanism design,\textquotedblright\ or
\textquotedblleft mechanism design with a competitive
fringe\textquotedblright ,\ e.g., \cite{phsk12}, \cite{tiro12}, \cite{cade15}%
, and \cite{fusk15}. In these papers, the platform is limited in the ability
to monopolize the market since the sellers have access to an outside option.
Our setting shares some of the same features, but in an oligopoly
environment where sellers compete for heterogeneous consumers. Furthermore,
the sellers choose their product menus understanding that customers arrive
through two different channels and that they have distinct information in
each channel.

At a broad level, this paper relates to information structures in
advertising auctions, e.g., \cite{bebm21}, and to nonlinear pricing, market
segmentation, and competition, e.g., \cite{bebm15}, \cite{bona11}, \cite%
{elgk20}, and \cite{yang22}. Finally, our analysis can be easily extended to
discuss self-preferencing by a monopoly platform. In this sense, our paper
also relates to \cite{hatw20}, \cite{muir21}, \cite{lam21}, \cite{lee21}, 
\cite{muso21}, and \cite{papp20}.

Finally, in parallel work, \cite{bebw23} compare auto-bidding through
managed advertising campaigns and data-augmented second-price auctions for
online advertising. They show that the managed campaign mechanism increases
the revenue of the digital platform and, with sufficient competition among
advertisers, it also increases consumer surplus.

\section{Model\label{sec:model}}

\paragraph{Sellers and Consumers}

We consider a digital platform and $J$ differentiated multiproduct sellers.
Each seller $j$ offers a product line (or menu) of quality differentiated
products. As in \cite{muro78}, each seller $j$ can produce a good of quality 
$q_{j}$ at a cost 
\begin{equation*}
c(q_{j})=q_{j}^{2}/2.
\end{equation*}%
There is a unit mass of consumers with single-unit demand. Each consumer is
described by a vector $\theta $\ of willingness-to-pay for quality for each
seller $j$'s products,%
\begin{equation*}
\theta =\left( \theta _{1},...,\theta _{j},...,\theta _{J}\right) \in
\lbrack \theta _{L},\theta _{H}]^{J},
\end{equation*}%
with $0\leq \theta _{L}<\theta _{H}<\infty $. We refer to the vector $\theta 
$ as the value profile of consumer $i$. Given a quality $q_{j}\,$\ offered
by seller $j$, the consumer receives a gross utility: 
\begin{equation*}
u\left( \theta ,q_{j}\right) =\theta _{j}\cdot q_{j}\text{.}
\end{equation*}

\paragraph{Information}

The value $\theta _{j}$\ of each consumer for each seller $j\ $are i.i.d.
across consumers and sellers with marginal distribution $F(\theta _{j})$ and
density $f(\theta _{j})$. Initially, each consumer has private information
with \emph{expectation} $m_{j}$\ (or expected value)\ about their true \emph{%
value} $\theta _{j}$. (We could alternatively refer to $m_{j}$ and $\theta
_{j}$ as interim and ex-post value, respectively.) The expectations $m_{j}$\
are assumed to be i.i.d. with marginal distribution $G\left( m_{j}\right) $
and density $g(m_{j})$. The distribution of values and expectations, $F$ and 
$G$, are implicitly related by an information structure. By \cite{blac51},
Theorem 5, there exists a signal $s$ that induces a distribution $G$\ of
expected values if and only if $F$ is a mean-preserving spread of $G$. We
recall that $F$ is defined to be a mean-preserving spread of $G$ if%
\begin{equation*}
\int_{v}^{\infty }F({t})dt\leq \int_{v}^{\infty }G({t})dt\text{, }\forall {v}%
\in \mathbb{R}_{+},
\end{equation*}%
with equality for $v=0$. If $F$ is a mean-preserving spread of $G$, we write 
$F\succ G$. (Conversely, $G$ is referred to as a mean-preserving contraction
of the given distribution $F$.)

\paragraph{On Platform}

A fraction $\lambda \in \lbrack 0,1]$ of all consumers uses the platform to
find a product.\footnote{%
We can endogenize the fraction of on-platform consumers $\lambda $ by
introducing heterogeneity in the cost and benefits of using the platform
(e.g., in the loss of privacy due to leaving \textquotedblleft
footprints\textquotedblright\ online).} The platform has access to extensive
data and knows each consumer's value profile $\theta $, while the sellers
only know the corresponding prior distribution $F$. The platform offers a
single sponsored link for a product $q_{j}$ sold by a seller $j$. The
platform uses a \emph{managed (advertising) campaign} mechanism to select
which seller and which product to sponsor to each consumer. This mechanism
captures all the salient features of automated bidding on real-world digital
advertising platforms.

In a managed campaign mechanism, the platform requests an advertising budget 
$t\in \mathbb{R}_{+}$ from each seller and announces a \emph{selection rule}%
. Any seller who pays the required budget submits quality and price
functions $q_{j}\left( \theta \right) $ and $p_{j}\left( \theta \right) $
representing the product and price that seller $j$ intends to advertise to
each consumer $\theta $. For each consumer value profile $\theta $, the
platform then advertises a single product $\left( q_{j}\left( \theta \right)
,p_{j}\left( \theta \right) \right) $ according to the selection rule. The
platform then chooses a sponsored seller according to the selection (or
steering) rule: 
\begin{equation}
\sigma :\Theta \times \mathbb{R}_{+}^{\Theta }\times \mathbb{R}_{+}^{\Theta
}\rightarrow \left[ J\right] \text{,}  \label{sigma}
\end{equation}%
and advertises the selected seller $j$'s product and price $q_{j}\left(
\theta \right) $ and $p_{j}\left( \theta \right) $.

Until Section \ref{friction}, we shall maintain the assumption that, upon
advertising product $q_{j}$ to consumer $\theta $ through the sponsored
link, the platform directly provides additional information to the consumer
that fully reveals their value $\theta _{j}$ for the advertised product.

\paragraph{Off Platform}

The remaining $1-\lambda $ consumers buy off the platform, e.g., from the
merchants' own websites or physical stores. Off the platform, the consumers
have their expectation $m$ and the sellers know the corresponding
distribution $G$. On the platform, there is extensive data and the platform
knows each consumer's value profile $\theta .$

The consumers who buy off the platform face positive (and arbitrarily small)
search costs beyond the first search, as in \cite{diamond71} and \cite%
{anre99}. The expectation $m$\ is private information of the consumer.
Therefore, seller $j$ elicits the consumer's private information through a
menu of (price, quality) pairs%
\begin{equation}
\{(\widehat{p}_{j}(m_{j}),\widehat{q}_{j}\left( m_{j}\right) )\}_{m_{j}\in %
\left[ \theta _{L},\theta _{H}\right] }  \label{menu0}
\end{equation}%
as in \cite{muro78} and \cite{MR84}. Throughout the paper, we thus use the
circumflex to distinguish off-platform variables from on-platform variables.
Importantly, the goods being sold are not experience or inspection goods: to
learn the vector $\theta $, consumers and sellers must gain access to the
platform's data.

After receiving seller $j$'s offer on the platform and learning their value $%
\theta _{j}$ for the advertised product $q_{j}$, each consumer can search
off the platform and use the newly gained information to select a product
from any seller. For example, the consumer can buy from the off-platform
schedule (\ref{menu0}) posted by the advertised seller $j$. Figure \ref%
{fig_model} summarizes the interaction between the agents and their actions
in our model. 
%TCIMACRO{\TeXButton{B}{\begin{figure}[htbp]\centering}}%
%BeginExpansion
\begin{figure}[htbp]\centering%
%EndExpansion
\includegraphics[width=.57\textwidth]{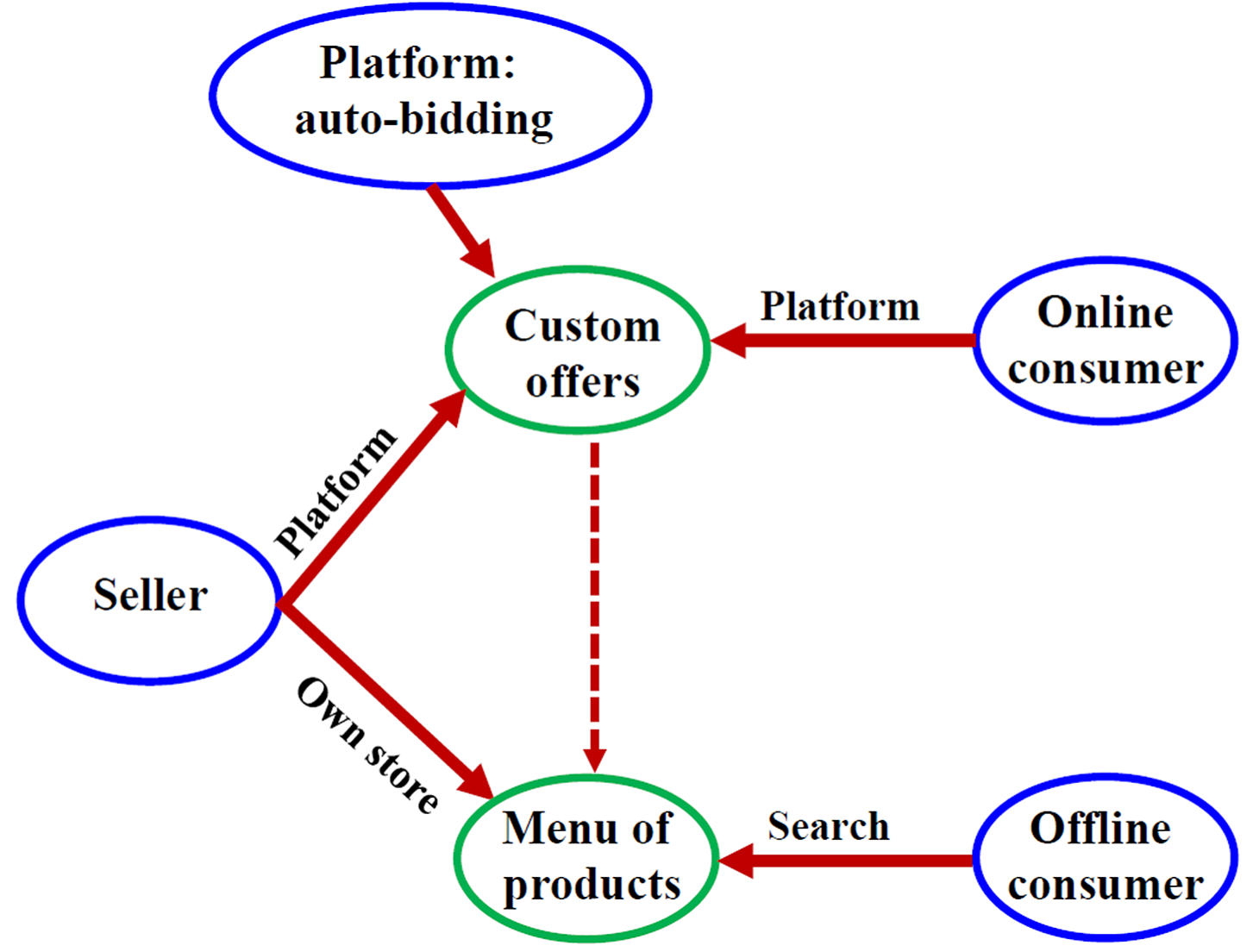}
%TCIMACRO{\TeXButton{caption}{\caption{On- and Off-Platform Matching}}}%
%BeginExpansion
\caption{On- and Off-Platform Matching}%
%EndExpansion
\label{fig_model}%
%TCIMACRO{\TeXButton{E}{\end{figure}}}%
%BeginExpansion
\end{figure}%
%EndExpansion

\paragraph{Timing and Equilibrium}

We consider simultaneous choices of (on- and off-platform) prices to capture
the great flexibility that algorithmic pricing offers the sellers both on
and off the platform. The timing of our game is as follows:

\begin{enumerate}
\item The platform announces a selection rule $\sigma $\ and requests an
advertising budget $t_{j}$ from each seller $j$.

\item Sellers simultaneously set off-platform products $\widehat{q}%
_{j}\left( m\right) $ and prices $\widehat{p}_{j}\left( m\right) $, choose
whether to submit the budget, and if so, what products $q_{j}\left( \theta
\right) $ and prices $p_{j}\left( \theta \right) $ to advertise.

\item The platform shows a single advertisement---a product $q_{j}(\theta )$
and a price $p_{j}(\theta )$---to each on-platform consumer according to the
announced selection rule.

\item The on-platform consumers learn their value $\theta _{j}$ for the
advertised seller; they can purchase the advertised product on the platform
or search off the platform.
\end{enumerate}

\begin{definition}[Symmetric Perfect Bayesian Equilibrium]
\quad \newline
We consider symmetric Perfect Bayesian Equilibria. The consumers have
symmetric beliefs over the sellers' off-platform menus both on and off the
equilibrium path.
\end{definition}

Thus, consumers expect sellers to play symmetric strategies on the
equilibrium path, and they continue to hold symmetric (though not
necessarily passive) beliefs over every seller's prices and qualities even
when they observe a deviation (either on or off the platform).\footnote{%
Symmetric beliefs off the path of play facilitate our exposition of the
intuition around Propositions \ref{prop_consider} and \ref{prop_knowntype}.
However, the results also hold under the passive beliefs refinement.}

\paragraph{Discussion}

We briefly comment on a few central aspects of the model. The managed
campaign mechanism has several important properties. First, sellers do not
acquire the platform's data, but they condition products and prices on the
platform's information about the consumer. In other words, the consumer's
value $\theta $ acts as a \emph{targeting category}. This corresponds to an
indirect sale of information as discussed in \cite{adpf90} and \cite{bebo19}%
. Second, fixed payments for advertising slots are similar to automated
bidding in ad auctions. Sellers submit a budget and upload the ads for the
products they wish to show to select consumers. Third, because each seller
can tailor the product offer to each consumer value, the platform creates an
opportunity for surplus extraction through \emph{product steering}, without
personalizing prices.

Several assumptions in the above mechanism can be easily relaxed. In
particular, we allow the sellers' schedules $q_{j}(\theta )$ and $%
p_{j}(\theta )$ to condition on the entire value $\theta $, and not just $%
\theta _{j}$. However, this additional flexibility will be redundant in
equilibrium. Our model of a single sponsored link is also simple in that the
platform sells both information and recognition to a single seller and a
specific product. In Section \ref{friction}, we extend our model to allow
for all brands and products to be present on the platform through organic
search results that advertise their off-platform offers. Finally, the direct
revelation of information to consumers captures the rich contextual detail
that some retail platforms provide to their users. Here, we assume that the
platform fully utilizes the informational advantage to generate surplus
through efficient product-consumer matching. In Section \ref{sec_infodesign}%
, we relax the assumption of perfect revelation of $\theta _{j}$, and we
study information design by the platform.

\section{A First Example: Single Seller\label{sec:exam}}

Before we begin with the analysis of the complete model, we illustrate some
of the central implications of our model with a simple example. The example
has a single seller (rather than many sellers) and binary values (rather
than a multidimensional continuum of values). In addition, the distribution
of consumer values and expected values is identical on and off the platform;
thus, $F=G$. The platform retains an informational advantage over the
sellers because it learns the value of the consumer that remains private
information off the platform.

The central result in this section (Proposition \ref{prop_single}) describes
the relationship between on-platform and off-platform pricing and quality
provision. Thus, even this very basic set-up sheds light on the fundamental
interaction between on-platform and off-platform allocations.

We consider a single seller that encounters a mass $\lambda $ of consumers
on the platform and a mass $1-\lambda $ of consumers off the platform.
Consumers can be of two values, $\theta \in \left\{ \theta _{L},\theta
_{H}\right\} $, each with probability $f\left( \theta \right) $. The
platform charges an advertising budget $t$ to the seller. With the provided
budget the seller earns the right to offer a personalized product to each
value of the consumer. However, each consumer on the platform can also shop
from the seller's own website (i.e., buy products the seller offers
off-platform). Thus, the consumer's option to \textquotedblleft
showroom\textquotedblright\ limits the seller's ability to price
discriminate.

If the seller accepts the platform's request of an advertising budget, it
offers a menu of products on the platform, which we describe in terms of the
product qualities $q\left( \theta \right) $\ and information rents $U\left(
\theta \right) $: 
\begin{equation}
\left\{ \left( q\left( \theta \right) ,U\left( \theta \right) \right)
\right\} _{\theta \in \left\{ \theta _{L},\theta _{H}\right\} },  \label{m1}
\end{equation}%
where the information rent is the net utility of the consumer on the
platform in equilibrium:%
\begin{equation*}
U\left( \theta \right) \triangleq \theta q\left( \theta \right) -p\left(
\theta \right) ,\ \ \theta \in \left\{ \theta _{L},\theta _{H}\right\} \text{%
.}
\end{equation*}%
The seller also offers a menu off the platform, denoted by 
\begin{equation}
\{(\widehat{q}\left( \theta \right) ,\widehat{U}\left( \theta \right)
)\}_{\theta \in \left\{ \theta _{L},\theta _{H}\right\} }.  \label{m2}
\end{equation}%
Throughout the paper, we thus use the circumflex to distinguish off-platform
variables from on-platform variables. The seller's profit is:%
\begin{equation*}
\max_{q,U}\sum_{\theta \in \left\{ \theta _{L},\theta _{H}\right\} }f\left(
\theta \right) [\lambda \left( \theta q\left( \theta \right) -q\left( \theta
\right) ^{2}/2-U\left( \theta \right) \right) +\left( 1-\lambda \right)
(\theta \widehat{q}\left( \theta \right) -\widehat{q}\left( \theta \right)
^{2}/2-\widehat{U}\left( \theta \right) )]\text{.}
\end{equation*}%
The seller maximizes profit subject to the individual rationality
constraints on and off the platform and to the incentive compatibility
constraints \emph{off }the platform because the consumers on the platform
receive a single targeted product offer. In addition, if the seller wants a
consumer to accept their targeted offer, the seller must induce the consumer
not to buy off the platform. The seller then faces the following new
\textquotedblleft showrooming\textquotedblright\ constraints:%
\begin{equation*}
U\left( \theta \right) \geq \widehat{U}\left( \theta \right) ,~\theta \in
\left\{ \theta _{L},\theta _{H}\right\} .
\end{equation*}%
In other words, each consumer $\theta $ must prefer to purchase on the
platform rather than to use the platform as a showroom and seek an
alternative quality-price pair off the platform.

It follows that the seller should offer the socially efficient quality
levels on the platform and that the showrooming constraint should bind, 
\begin{equation}
q\left( \theta \right) =\theta ~\text{and~}U\left( \theta \right) =\widehat{U%
}\left( \theta \right) ,~\theta \in \left\{ \theta _{L},\theta _{H}\right\} .
\label{pa}
\end{equation}%
We now characterize the quality levels off the platform. As usual, the
equilibrium menu does not distort at the top $\left( \widehat{q}\left(
\theta _{H}\right) =\theta _{H}\right) $ and offers no rents at the bottom $(%
\widehat{U}\left( \theta _{L}\right) =0)$. Furthermore, the incentive
compatibility constraint binds for the high value. With these preliminary
results, the seller's objective can be written as%
\begin{eqnarray}
&&\max_{q,U}\left[ \lambda (f\left( \theta _{L}\right) \theta
_{L}^{2}/2+f\left( \theta _{H}\right) (\theta _{H}^{2}/2-\widehat{U}\left(
\theta _{H}\right) ))\right.  \label{ss_prob} \\
&&\left. +\left( 1-\lambda \right) (f\left( \theta _{L}\right) \left( \theta
_{L}\widehat{q}\left( \theta _{L}\right) -\widehat{q}\left( \theta
_{L}\right) ^{2}/2\right) +f\left( \theta _{H}\right) (\theta _{H}^{2}/2-%
\widehat{U}\left( \theta _{H}\right) ))\right] \text{ }  \notag
\end{eqnarray}%
subject to the constraint 
\begin{equation}
\widehat{U}\left( \theta _{H}\right) =\left( \theta _{H}-\theta _{L}\right) 
\widehat{q}\left( \theta _{L}\right) .  \label{ICH}
\end{equation}%
From this expression, it is immediate that the provision of quality to the
low value off the platform is doubly costly for the seller: it forces the
seller to lower the price for the high value off the platform, and it also
forces lower prices on the platform.

\begin{proposition}[Single Seller and Binary Values]
\label{prop_single}\qquad \newline
The optimal off-platform menu of products for the seller is:%
\begin{eqnarray}
\widehat{q}\left( \theta _{L}\right) &=&\max \left\{ 0,\theta _{L}-\frac{%
f\left( \theta _{H}\right) }{f\left( \theta _{L}\right) }\left( \theta
_{H}-\theta _{L}\right) \left( 1+\frac{\lambda }{1-\lambda }\right) \right\}
,  \label{qlow} \\
\widehat{q}\left( \theta _{H}\right) &=&\theta _{H}.  \notag
\end{eqnarray}
\end{proposition}

We relegate the formal proof of all our results to the Appendix. Let us now
compare the optimal menu with the classic nonlinear pricing solution as in 
\cite{muro78}, which corresponds to the case $\lambda =0$. In that case, we
would have%
\begin{eqnarray}
\widehat{q}\left( \theta _{L}\right) &=&\max \left\{ 0,\theta _{L}-\frac{%
f\left( \theta _{H}\right) }{f\left( \theta _{L}\right) }\left( \theta
_{H}-\theta _{L}\right) \right\} ,  \label{qold} \\
\widehat{q}\left( \theta _{H}\right) &=&\theta _{H}.  \notag
\end{eqnarray}%
Proposition \ref{prop_single} indicates an additional opportunity cost of
serving the low value off the platform. Indeed, it is not difficult to find
parameters (e.g., $\lambda $ large enough) for which the quality level
on-platform as in (\ref{qold}) is strictly positive but the quality
off-platform as in (\ref{qlow}) is zero. Thus, without a platform, the
seller would offer a low-quality product to the low value. However, for a
sufficiently large platform, the low value is only offered a product on the
platform, where the seller can make a different personalized offer to the
high value. This is not the case off the platform, where the seller prefers
to forego sales of the low product, to sell product $q\left( \theta
_{H}\right) =\theta _{H}$ at a higher price on both channels. Indeed, when $%
\widehat{q}\left( \theta _{L}\right) =0$, no consumer value receives any
rent on or off the platform.

Finally, to determine the optimal advertising budget $t^{\ast }$, we need to
consider what the on-platform consumers would do if the seller did not
advertise. If these consumers can buy off the platform, the optimal
advertising budget extracts the seller's extra profit relative to offering
the menu in (\ref{qold}) to all consumers. If these consumers were not to
buy at all, the seller's outside option is scaled by a factor $1-\lambda $,
and the optimal advertising budget is correspondingly higher.

The environment with many sellers and many values that we consider next
requires a richer analysis. With many sellers, we must consider how the
information on the platform impacts the search behavior off the platform
(Proposition \ref{prop_consider}). In turn, this determines the shape of the
menu offered by the sellers in the presence of competition (Proposition \ref%
{menu}) and the nature of the revenue-maximizing mechanism for the platform
(Proposition \ref{optimalmech}).

\section{Managed Campaign and Showrooming\label{sec:search}}

We now analyze the environment with many sellers and a multidimensional
continuum of values as introduced in Section \ref{sec:model}. Our objective
is to establish the equilibrium patterns of consumer search induced by the
informational advantage of the platform.\ (This advantage is captured by the
distinction between values $\theta _{j}$ with distribution $F$ on the
platform and expected values $m_{j}$ with distribution $G$ off the platform.)

We uncover the following tradeoff: off the platform, each seller $j$ faces
those\ consumers who value their product the most based on their expected
value $m$. However, trade takes place under asymmetric information, i.e.,
the seller must elicit the consumers' willingness to pay. In contrast, the
platform enables consumers and sellers to interact under symmetric
information. Thus, the sellers are willing to pay for the right to make a
personalized offer to each consumer they are matched to under the platform's
managed campaign mechanism. However, if the winning seller wants a consumer
to accept their personalized offer, this seller must induce the consumer not
to buy from the off-platform store, i.e., not to use the platform for
\textquotedblleft showrooming.\textquotedblright\ Thus, the seller's ability
to product steer and price discriminate on the platform is limited by the
presence of the off-platform channel.

\subsection{Managed Campaign}

As we introduced in Section \ref{sec:model}, the platform designs a managed
advertising campaign to select a sponsored product. This mechanism is
defined formally as follows.

\begin{definition}[Managed Advertising Campaign]
\quad \newline
In a managed campaign mechanism:

\begin{enumerate}
\item the platform requests an advertising budget $t\in \mathbb{R}_{+}$ from
each seller as a take-it-or-leave-it offer;

\item each participating seller $j$ submits schedules $q_{j}:\Theta
\rightarrow \mathbb{R}_{+}$ and $p_{j}:\Theta \rightarrow \mathbb{R}_{+}$,
which represent the quality and price seller $j$ intends to advertise to
each consumer $\theta $;

\item for each consumer value profile $\theta \in \left[ \theta _{L},\theta
_{H}\right] ^{J}$, the platform chooses a sponsored seller according to the
selection mapping 
\begin{equation*}
\sigma :\Theta \times \mathbb{R}_{+}^{\Theta }\times \mathbb{R}_{+}^{\Theta
}\rightarrow \left[ J\right] \text{,}
\end{equation*}%
and advertises the selected seller $j$'s product and price $q_{j}\left(
\theta \right) $ and $p_{j}\left( \theta \right) .$\bigskip
\end{enumerate}
\end{definition}

Throughout the paper, we focus on a specific selection rule, namely the one
that matches consumers and products efficiently. Indeed, we establish the
revenue optimality of this mechanism in Proposition \ref{optimalmech}.

\begin{definition}[Efficient Steering]
\quad \newline
For each value profile $\theta $, the platform chooses the seller that
maximizes the social value of the match 
\begin{equation*}
\sigma ^{\ast }\left( \theta ,q\left( \theta \right) ,p\left( \theta \right)
\right) =\arg \max_{j}\left[ \theta _{j}q_{j}\left( \theta \right)
-c(q_{j}\left( \theta \right) )\right] \text{,}
\end{equation*}%
among all the sellers that participate in the mechanism.
\end{definition}

We now derive the equilibrium choice patterns when the platform steers
consumers to products efficiently.

\subsection{Choice Patterns}

We begin with the off-platform consumers. These consumers (which have mass $%
1-\lambda $) face positive search costs beyond the first seller. As a
result, in any symmetric equilibrium, a consumer with expected value $m$
visits only the seller who offers the highest expected value:%
\begin{equation*}
j^{\left( 1\right) }=\arg \max_{j}m_{j}.
\end{equation*}%
This result does not depend on the magnitude of the search costs as
established famously by \cite{diamond71}. Moreover, if the platform has a
strict informational advantage ($F\succ G$), the on-platform consumers
(which have mass $\lambda $) infer that the advertised seller $j^{\ast }$
maximizes their willingness to pay, i.e., $\theta _{j^{\ast
}}=\max_{j}\theta _{j}.$ Because these consumers expect symmetric menus off
the platform and the information rent function associated with those menus
is strictly increasing, these consumers consider products offered by the
advertised seller $j^{\ast }$ only.

\begin{proposition}[Consideration Sets]
\label{prop_consider}\qquad \newline
Every on-platform consumer $\theta $ compares the advertised seller's
on-platform offer $(p_{j^{\ast }}(\theta ),q_{j^{\ast }}(\theta ))$ only
with the corresponding off-platform offer $(\widehat{p}_{j^{\ast }}(\theta
_{j^{\ast }}),\widehat{q}_{j^{\ast }}(\theta _{j^{\ast }}))$.
\end{proposition}

The platform augments the expectation of each consumer with additional data
that lead to a revision from the expected value $m$ to the (true) value $%
\theta $. Figure \ref{fig_pattern} below illustrates the choice behavior by
a consumer whose expected value $m$\ rank the sellers differently relative
to (true) value $\theta $. The consumer has two possible choices, thus $J=2$
and the expected value $m$ suggests that seller $1$ offers a higher value,
thus $m_{1}>m_{2}$. Now suppose that on the platform the consumer is shown
an advertisement by seller, $j=2$. Indeed, the platform reveals to the
consumer the value $\theta _{2}$, and thus they will infer that $\theta
_{2}>\theta _{1}$. Therefore, by the above Proposition, the consumer either
accepts seller $2$'s offer or shops off the platform from seller $2$, but
now with full knowledge of their value $\theta _{2}$.

%TCIMACRO{\TeXButton{B}{\begin{figure}[htbp]\centering}}%
%BeginExpansion
\begin{figure}[htbp]\centering%
%EndExpansion
\includegraphics[width=.8\textwidth]{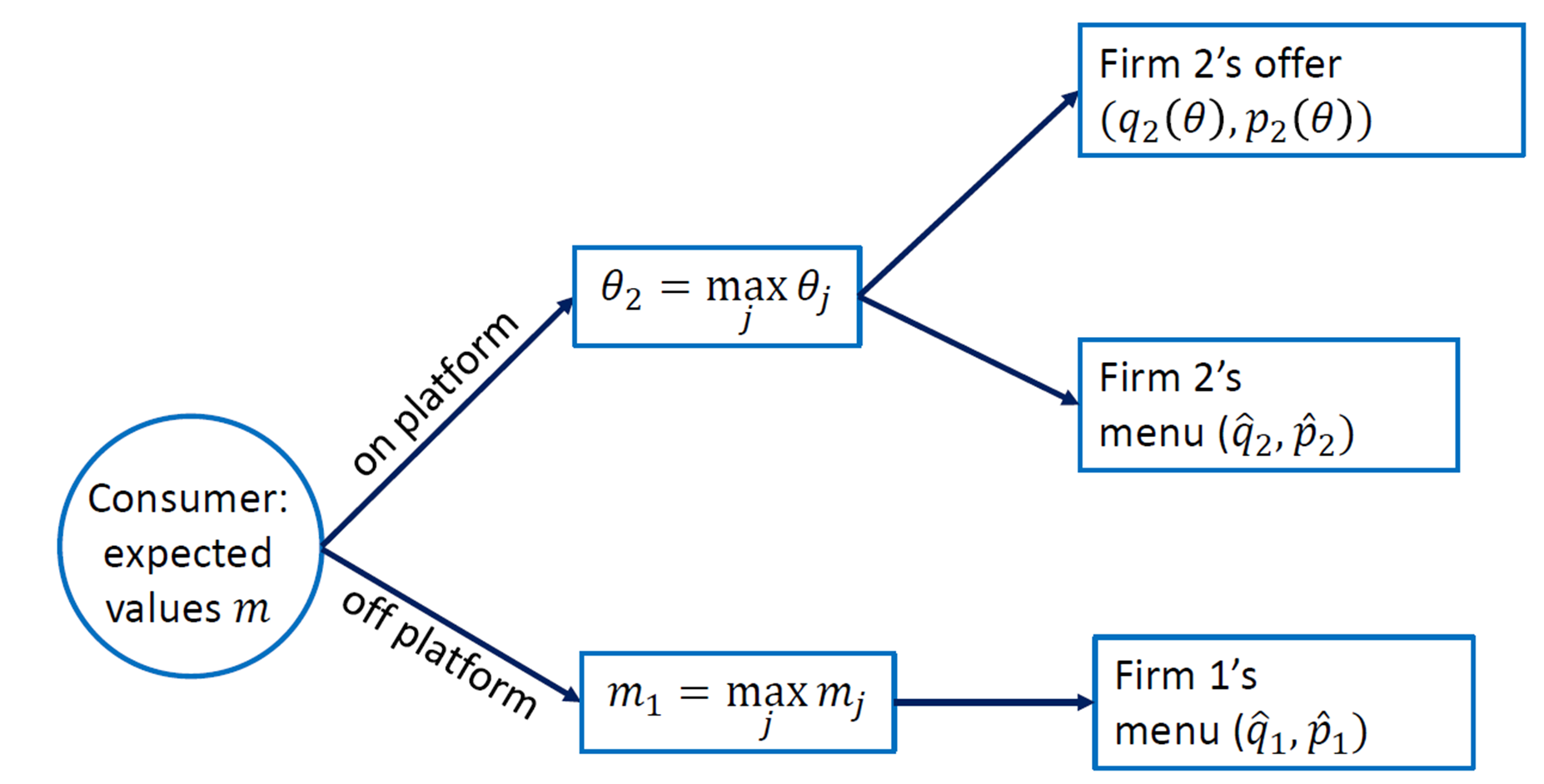}
%TCIMACRO{%
%\TeXButton{caption}{\caption{Possible Consumer Consideration Sets}}}%
%BeginExpansion
\caption{Possible Consumer Consideration Sets}%
%EndExpansion
\label{fig_pattern}%
%TCIMACRO{\TeXButton{E}{\end{figure}}}%
%BeginExpansion
\end{figure}%
%EndExpansion

An important implication of Proposition \ref{prop_consider} is that every
consumer will (possibly incorrectly) buy from a competing seller if they do
not see their favorite seller's ad. Thus, every seller realizes that
participating in the platform's mechanism is necessary to access \emph{any}
of the on-platform consumers.

Indeed, as the platform has better information than the consumers, any
symmetric equilibrium of the game is outcome-equivalent to a simpler model
where each seller has a group of $(1-\lambda )/J$ customers who only
consider their brand. These consumers buy a product variety that depends on
their expected value $m_{j}$, which is distributed according to $G^{J}$. The
remaining $\lambda $ consumers are currently not loyal shoppers for any
brand (i.e., no seller is in their consideration set), but they become aware
of a buying opportunity upon seeing an ad. In this case, they can buy from
the only seller in their consideration set---the seller with the sponsored
link or advertisement.\footnote{\cite{memupa23} consider the impact of
exogenous information at each stage of a consumer's sequential search
process, while \cite{chen22} studies the evolution of a consumers'
consideration set under the equilibrium advertising levels.}

This alternative interpretation in terms of endogenous consideration sets
requires the platform to hold an (arbitrarily small) informational advantage
relative to the consumers. In Section \ref{friction}, we show that, without
an informational advantage, the platform does not control the consumers'
outside options. Instead, the consumers' expected value fully determines
which seller they would visit off the platform.

\subsection{Showrooming}

The platform generates surplus by matching each $\theta $
to the product $q_{j}\left( \theta \right) $ that generates the largest
match value. As information is symmetric between the consumer and the
selected seller, the seller can extract a substantial share of the created
social surplus. To wit, the extraction of the surplus does not occur through
personalized price discrimination but through product steering (thus a form
of second-degree price discrimination). The only limit on surplus extraction
by the advertising seller is given by the \textquotedblleft showrooming
constraint,\textquotedblright\ which is a necessary condition for seller $j$
to make a sale on the platform:%
\begin{equation}
\theta _{j}\cdot q_{j}(\theta )-p_{j}(\theta )\geq \max_{m_{j}}\left[ \theta
_{j}\cdot \widehat{q}_{j}(m_{j})-\widehat{p}_{j}(m_{j})\right] \text{ for
all }\theta _{j}.  \label{show1}
\end{equation}%
Because seller $j$ offers an incentive-compatible menu off the platform,
each on-platform consumer would also report their value truthfully if
shopping off-platform. By Proposition \ref{prop_consider}, we know the
consumer chooses between two products by the same seller, and the
showrooming constraint (\ref{show1}) reduces to 
\begin{equation}
U_{j}\left( \theta \right) :=\theta _{j}q_{j}(\theta )-p_{j}(\theta ))\geq
\theta _{j}\widehat{q}_{j}(\theta _{j})-\widehat{p}_{j}(\theta _{j}))=:%
\widehat{U}_{j}\left( \theta \right) .  \label{show2}
\end{equation}

The showrooming constraint prevents the selected seller from extracting the
entire surplus of the on-platform consumers. Because the on-platform
transaction takes place under symmetric information, it is optimal for each
seller to offer a single product to each consumer $\theta $ at the socially
efficient quality level%
\begin{equation*}
q_{j}^{\ast }\left( \theta \right) =\theta _{j}.
\end{equation*}%
The socially efficient quality provision maximizes both the profit from the
ad and the probability of being chosen by the platform. Similarly, it is
optimal for each seller to offer the consumer a discount that satisfies the
showrooming constraint (\ref{show2}) with equality. Thus, despite the
flexibility awarded by the platform, the quality and utility offered
on-platform by firm $j$ are a function of $\theta _{j}$ only.

Finally, if seller $j$ offers the off-platform menu $(\widehat{p}_{j},%
\widehat{q}_{j})$ with the associated rent function $\widehat{U}_{j}$, the
on-platform profit from a consumer with value profile $\theta $ is given by%
\begin{equation}
\pi _{j}(\theta ,\widehat{U}_{j})=\left\{ 
\begin{tabular}{ll}
$\theta _{j}^{2}/2-\widehat{U}_{j}(\theta _{j}),$ & if $\theta
_{j}>\max_{k\neq j}\theta _{k};$ \\ 
$0,$ & otherwise.%
\end{tabular}%
\right.  \label{prof_onplat}
\end{equation}

\section{Equilibrium Product Lines\label{sec:epl}}

We now characterize the symmetric equilibrium menus off the platform and
trace their implications for on-platform quantities and prices. We can then
analyze the expected consumer surplus on the platform and off the platform.
Finally, we establish that the socially efficient steering mechanism is the
revenue-maximizing managed campaign for the platform.

By Proposition \ref{prop_consider}, in any symmetric equilibrium of our
model, no consumer (off-platform or on-platform) searches past the first
seller on the equilibrium path. Combining the off-platform profit with the
on-platform profit (\ref{prof_onplat}), each seller's maximization problem
can be written as follows.

\begin{align}
\Pi _{j}^{\ast }=\max_{\widehat{q}_{j},\widehat{U}_{j}}\,& (1-\lambda
)\int_{\theta _{L}}^{\theta _{H}}[m_{j}\widehat{q}_{j}(m_{j})-\widehat{q}%
_{j}(m_{j})^{2}/2-\widehat{U}_{j}(m_{j})]G^{J-1}(m_{j})g(m_{j})\text{d}m_{j}
\label{obj1} \\
& +\lambda \int_{\theta _{L}}^{\theta _{H}}[\theta _{j}^{2}/2-\widehat{U}%
_{j}(\theta _{j})]F^{J-1}(\theta _{j})f(\theta _{j})\text{d}\theta _{j}\text{%
,}  \notag \\
\text{s.t. \qquad }& \widehat{U}_{j}(m_{j})\geq 0\text{,}  \tag{IR} \\
& \widehat{U}_{j}^{\prime }(m_{j})=\widehat{q}_{j}(m_{j})\text{.}  \tag{IC}
\end{align}

We observe that the objective function of each seller takes into account the
competition across the seller's own sales channels. Each seller $j$\
generates an expected value $m_{j}$ for the consumer off the platform with
density $g\left( m_{j}\right) $ but makes a sale only if $j$ has a higher
expected value than the remaining $J-1$ sellers, which happens with
probability $G^{J-1}(m_{j})$. A similar expression involving $f\left( \theta
_{j}\right) $ and $F^{J-1}(\theta _{j})$ holds for the on-platform revenue.
Thus, the above pricing and revenue formulas take into account the
competition among the $J$ sellers by taking expectation with respect to the
highest-order statistics.

Because we have assumed $m_{j}\in \left[ \theta _{L},\theta _{H}\right] ,$
we state all our results with $\theta _{j}$ as the argument, letting the
distributions $F$ and $G$ indicate whether we refer to on- or off-platform
variables. Maximizing (\ref{obj1}) over rent and quality functions $\widehat{%
U}_{j}$ and $\widehat{q}_{j}$, we obtain the following characterization of
the optimal menus.\footnote{%
As we have ignored the monotonicity constraint in problem (\ref{obj1}),
Proposition \ref{menu} applies to the case where the distributions $F$ and $%
G $ are sufficiently regular. If they are not, i.e., if the function $\hat{q}%
_{j}^{\ast }\left( \theta _{j}\right) $ in (\ref{q0s}) is not monotone, then
the equilibrium quality schedule is given by the \textquotedblleft
ironed\textquotedblright\ version of (\ref{q0s}).}

\begin{proposition}[Equilibrium Menus with Efficient Steering]
\label{menu}\strut \vspace*{-0.1in}

\begin{enumerate}
\item The unique symmetric equilibrium quality levels are given by: 
\begin{align}
q_{j}^{\ast }(\theta _{j})& =\theta _{j},  \label{q00s} \\
\widehat{q}_{j}^{\ast }(\theta _{j})& =\max \left\{ 0,\theta _{j}-\frac{%
1-\lambda F^{J}(\theta _{j})-(1-\lambda )G^{J}(\theta _{j})}{(1-\lambda
)JG^{J-1}(\theta _{j})g(\theta _{j})}\right\} .\vspace*{-0.1in}  \label{q0s}
\end{align}

\item The consumer's information rents are identical on- and off-platform:%
\begin{equation}
U_{j}^{\ast }(\theta _{j})=\widehat{U}_{j}^{\ast }(\theta _{j})=\int_{\theta
_{L}}^{\theta _{j}}\widehat{q}_{j}^{\ast }(x)\text{d}x.\vspace*{-0.1in}
\label{U0s}
\end{equation}
\end{enumerate}
\end{proposition}

The equilibrium quality provision on and off the platform have several
intuitive properties. First, the efficient quality is sold to each consumer $%
i$ on the platform, on the basis of their favorite seller, i.e., $%
\max_{j}\{\theta _{j}\}$. Second, matching is inefficient off the platform
because it is based on imperfect information, i.e., on expected value $m$
instead of value $\theta $.

The consumer's private information off the platform requires the sellers to
resolve the efficiency vs. rent extraction trade-off. The information rents
of each value $m_{j}$ are as usual increasing in the quality level provided
to all lower values. To resolve this trade-off, each seller $j$ could offer
the \cite{muro78} tariff for the distribution of off-platform consumer
values $G^{J}(m_{j})$, which is the distribution of the highest order
statistic out of $J$ variables $m_{j}$. However, any information rent $%
\widehat{U}(m)$ provided to the off-platform consumers has an additional
shadow cost: it makes buying off-platform more attractive for the
on-platform consumers too. As we saw, by leaving positive rents for the
consumers off the platform, each seller must also provide rents on the
platform: 
\begin{equation*}
U_{j}^{\ast }(\theta _{j})=\widehat{U}_{j}^{\ast }(\theta _{j})>0\text{ iff }%
\widehat{q}_{j}^{\ast }(\theta _{j})>0.
\end{equation*}%
Conversely, by limiting the off-platform rents, the seller is able to
capture a greater share of the efficient social surplus that personalized
on-platform offers generate.

Because of the shadow cost of showrooming, the off-platform quality schedule 
$\widehat{q}$ is further distorted downward. In particular, we can rewrite
the equilibrium off-platform qualities (\ref{q0s}) in Proposition \ref{menu}
as%
\begin{equation}
\widehat{q}_{j}^{\ast }(\theta _{j})=\underset{\text{\cite{muro78} quality}}{%
~\underbrace{\theta _{j}-\frac{1-G^{J}(\theta _{j})}{JG^{J-1}(\theta
_{j})g(\theta _{j})}}}-~\frac{\lambda }{1-\lambda }\frac{1-F^{J}(\theta _{j})%
}{JG^{J-1}(\theta _{j})g(\theta _{j})},  \label{q0ss}
\end{equation}%
where the first two terms identify the optimal quality level for the
distribution of values $G^{J}(\theta _{j})$. The last term captures the
intuition that any rent given off-platform to value $\theta _{j}$ must also
be given to all higher values \emph{on }the platform.

%TCIMACRO{\TeXButton{B}{\begin{figure}[htbp]\centering}}%
%BeginExpansion
\begin{figure}[htbp]\centering%
%EndExpansion
\includegraphics[width=.5\textwidth]{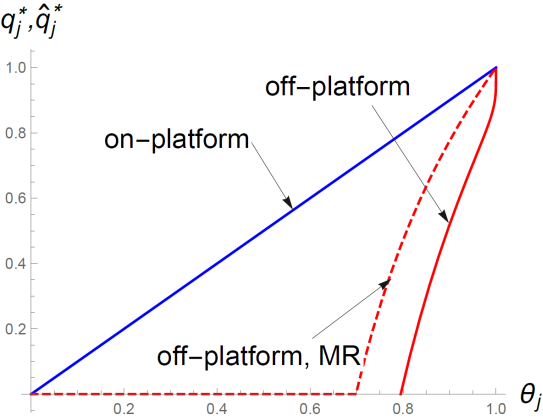}
%TCIMACRO{%
%\TeXButton{caption}{\caption{Match Values and Qualities On-Platform vs Off-Platform, $\la=1/2,J=5,G(m_j)=m_j,F(\ta_j)=\textrm{Beta}(\ta_j,1/4,1/4)$}}}%
%BeginExpansion
\caption{Match Values and Qualities On-Platform vs Off-Platform, $\la=1/2,J=5,G(m_j)=m_j,F(\ta_j)=\textrm{Beta}(\ta_j,1/4,1/4)$}%
%EndExpansion
\label{qmr}%
%TCIMACRO{\TeXButton{E}{\end{figure}}}%
%BeginExpansion
\end{figure}%
%EndExpansion

Figure \ref{qmr} displays the equilibrium off-platform quality schedule, the
socially efficient allocation, and the monopoly benchmark of \cite{muro78}.
The consumer's expected values $m_{j}$ are uniformly distributed, and their
values $\theta _{j}$ follow a Beta distribution.

The formulation of the optimal off-platform menu (\ref{q0ss}) allows us to
establish several intuitive properties of the equilibrium. First, each value 
$\theta $ receives a higher quality level, namely the socially efficient
allocation and pays a higher price on the platform than off the platform.
However, while each value receives a better product at a higher price, each
quality level $q$ is sold at a lower price on the platform. Thus, let us
define the equilibrium price for a given quality $q$ on- and off-platform as 
$p\left( q\right) $ and $\widehat{p}\left( q\right) $, respectively. We then
find that $p\left( q\right) \leq \widehat{p}\left( q\right) $ for all $q$.
In other words, each seller is forced to introduce \textquotedblleft
on-platform only\textquotedblright\ discounts due to the threat of
showrooming.

Figure \ref{nonlin} displays the nonlinear pricing schedules, namely the
price $p_{j}\left( q_{j}\right) $ for every offered quality $q_{j}$ under
the same parameters as in Figure \ref{qmr}. Note that for a set of low
values, namely those below $8/10$, the nonlinear tariff is equal to the
gross surplus generated by the efficient quality (i.e., $p_{j}\left(
q_{j}\right) =q_{j}^{2}$). By contrast, values above $8/10$ receive a
positive rent off-platform, and hence on the platform the price is below the
gross surplus.

%TCIMACRO{\TeXButton{B}{\begin{figure}[htbp]\centering}}%
%BeginExpansion
\begin{figure}[htbp]\centering%
%EndExpansion
\includegraphics[width=.5\textwidth]{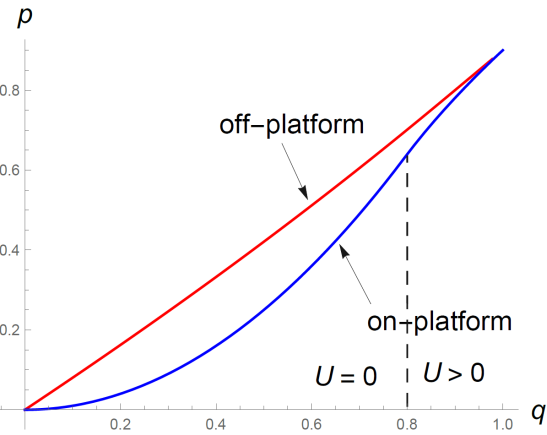}
%TCIMACRO{%
%\TeXButton{caption}{\caption{Nonlinear tariffs: every variety $q$ is sold at a lower price on-platform.}}}%
%BeginExpansion
\caption{Nonlinear tariffs: every variety $q$ is sold at a lower price on-platform.}%
%EndExpansion
\label{nonlin}%
%TCIMACRO{\TeXButton{E}{\end{figure}}}%
%BeginExpansion
\end{figure}%
%EndExpansion

Our results illustrate how advertising platforms that run managed campaigns
face very different incentives than retail platforms that charge
proportional transaction fees. In the case of Amazon, for example, merchants
may want consumers to showroom to avoid the platform's variable fees. In the
case of our advertising platform, ex-ante, fixed-price contracts with the
sellers eliminate the need for most-favored-nation clauses.

\paragraph{Consumer Surplus}

An implication of Proposition \ref{menu} is that, on aggregate, consumer
surplus is higher on the platform than off the platform. Indeed, for each
value $\theta $, we have%
\begin{equation*}
\widehat{U}_{j}^{\ast }(\theta _{j})=U_{j}^{\ast }(\theta _{j}).
\end{equation*}%
However, we also know that $F\succ G$. This implies%
\begin{equation*}
\mathbb{E}_{F^{J}}[U_{j}^{\ast }(\theta _{j})]>\mathbb{E}_{G^{J}}[\widehat{U}%
_{j}^{\ast }(\theta _{j})],
\end{equation*}%
because the highest order statistics satisfy $\mathbb{E}_{F^{J}}\left[
\theta _{j}\right] >\mathbb{E}_{G^{J}}\left[ \theta _{j}\right] $ and
incentive compatibility requires the function $\widehat{U}_{j}^{\ast }$ to
be increasing and convex. Thus, at the equilibrium prices, every consumer
would rather be on the platform (ex ante) than off the platform. A stronger
result is that, holding prices fixed, the consumer would like the platform
to have as precise information as possible about their value, which enables
better matching of products to preferences. However, the consumer does not
necessarily benefit from the presence of the platform \emph{in equilibrium}.
Indeed, Proposition \ref{compstat} considers the effects of a larger
platform $\left( \lambda \right) $ and finds that \emph{all }consumers are
worse off as the share $\lambda $ of consumers who shop on the platform
increases.

\paragraph{Platform Revenue}

To examine the implications for the sellers' profit and the platform's
revenue, we characterize the advertising budgets the platform can demand in
equilibrium under a managed campaign mechanism.

We first define each seller $j$'s outside option $\overline{\Pi }_{j}\left(
\sigma \right) $ as the profit seller $j$ can obtain if they do not
participate in a managed campaign with selection rule $\sigma $. Similarly,
define each seller's equilibrium profits (gross of the advertising budget)
as $\Pi _{j}^{\ast }\left( \sigma \right) .$ Because sellers are homogeneous
ex ante, the platform can then request an advertising budget equal to%
\begin{equation*}
t^{\ast }\left( \sigma \right) \triangleq \Pi _{j}^{\ast }\left( \sigma
\right) -\overline{\Pi }_{j}\left( \sigma \right) .
\end{equation*}

Now consider the efficient selection rule $\sigma ^{\ast }$ we have examined
so far. In the efficient steering managed campaign, consumers follow the
platform's recommendation on and off the equilibrium path (Proposition \ref%
{prop_consider}). Each seller's outside option then consists of the profit
level achievable from the off-platform consumers only, i.e.,%
\begin{equation}
\overline{\Pi }_{j}\left( \sigma ^{\ast }\right) =\max_{\widehat{q},\widehat{%
U}}\,(1-\lambda )\int_{\theta _{L}}^{\theta _{H}}\left[ \theta _{j}\widehat{q%
}(\theta _{j})-\widehat{q}(\theta _{j})^{2}/2-\widehat{U}(\theta _{j})\right]
G^{J-1}(\theta _{j})\text{d}G(\theta _{j}).  \label{eq_ou}
\end{equation}

Given the equilibrium profit levels $\Pi _{j}^{\ast }\left( \sigma ^{\ast
}\right) $ defined in (\ref{obj1}), the platform then requests the following
advertising budget:%
\begin{equation}
t^{\ast }\left( \sigma ^{\ast }\right) =\Pi _{j}^{\ast }\left( \sigma ^{\ast
}\right) -\overline{\Pi }_{j}\left( \sigma ^{\ast }\right) .  \label{tstar}
\end{equation}%
Another way to interpret the advertising budget is the following: the
sellers are willing to give up all the on-platform profit to participate,
but need to be compensated for distorting the off-platform menus away from
the monopoly benchmark.

We now consider the platform's problem when announcing a selection rule,
i.e.,%
\begin{equation*}
\max_{\sigma }\left[ t^{\ast }\left( \sigma \right) \right] ,
\end{equation*}
and we show that the efficient-steering mechanism indeed solves this problem.%
\footnote{%
Consistent with this result, \cite{bebw23} show that the optimal managed
campaign improves the platform's revenue relative to running auctions for
targeted advertising slots. Furthermore, the experimental analysis of \cite%
{decaro23} shows that, when a platform provides less detailed information to
the bidders' algorithms, its revenues are \textquotedblleft substantially
and persistently higher.\textquotedblright}

\begin{proposition}[Optimality of Efficient Steering]
\label{optimalmech}\qquad \newline
The efficient-steering managed campaign mechanism maximizes the platform's
profit.
\end{proposition}

The proof of this result establishes that the efficient-steering managed
campaign attains an upper bound on the platform's profit. In particular, the
equilibrium profits of the sellers $\Pi _{j}^{\ast }\left( \sigma ^{\ast
}\right) $ coincide with the vertically integrated benchmark where the
platform controls all sellers' menus on both sales channels. Thus, this
mechanism maximizes the sellers' and the platform's joint profits across the
on- and off-platform markets. Moreover, the seller's outside option $%
\overline{\Pi }_{j}\left( \sigma \right) $ is bounded from below, for all
managed campaigns $\sigma $, by the profits $\overline{\Pi }_{j}\left(
\sigma ^{\ast }\right) $ a seller can obtain through the optimal menu for
the off platform consumers only. Because the advertising budget extracts
each seller's surplus over and above their exogenous outside option, this
mechanism maximizes the platform's revenue. Importantly, the
efficient-steering managed campaign relies only on advertising and steering
through sponsored products. In particular, the optimal revenue can be
attained without levying commission or transaction fees on sellers or
consumers.\footnote{%
In the presence of complete information by the platform, it suffices that
the platform offers a single sponsored product rather many sponsored product
slots. We suspect that in richer environments where consumers have some
independent private information, many sponsored links would optimally screen
for this additional information.}

\section{Value of Information and Privacy \label{friction}}

In this section, we explore the platform's bargaining power by examining the
role of its informational advantage and the implications of privacy
policies. We start by removing the informational advantage of the platform
in Section \ref{symmetric}. Instead, we assume that every on-platform
consumer learns their entire value profile $\theta $, not just their value
for the sponsored seller. One possible reason for this could be that reviews
and recommendations are available on the platform and online more generally.

Next, in Section \ref{sec_compdesign}, we examine the role of price
information. We consider the provision of organic search links by the
platform that enable consumers to learn about all off-platform prices and
products.

In Section \ref{privacy}, we introduce privacy policies that safeguard the
consumers' information from the sellers. We consider cohort-based privacy
protection where the platform informs the sellers only about the consumer's
ranking of the sellers, while disclosing the exact value for the sponsored
seller to the consumer. Consequently, the platform targets ads at the level
of a \emph{cohort }of consumers, and each consumer within a cohort have the
same preference ranking over the $J$ sellers.\footnote{%
This is in line with the recent Google Privacy Sandbox proposals to replace
third-party cookies.}

Lastly, in Section \ref{sec_infodesign}, we examine whether revealing the
full value for the sponsored seller to the consumer maximizes the platform's
revenue. We introduce information design in our managed campaign mechanism
and provide conditions under which full or partial information revelation is
optimal.

\subsection{Symmetric Information\label{symmetric}}

To assess the value of the platform's information advantage, we now assume
all consumers who visit the platform learn their \emph{entire value profile }%
$\theta $ (i.e., not just their value for the sponsored seller). The
consumers off the platform remain imperfectly informed with expected value
profile $m$. In Proposition \ref{prop_knowntype}, we establish that the
ensuing symmetric information limits the platform's ability to steer the
consumers' search behavior and reduces the advertising budget the platform
can request from the sellers.

\begin{proposition}[Symmetric Information]
\label{prop_knowntype}\qquad \newline
With complete information about $\theta $ for all on-platform consumers, the
equilibrium quality levels on and off the platform remain as in Proposition %
\ref{menu}. But the equilibrium advertising budget $t^{\ast }$ is strictly
lower relative to when the platform has exclusive information about $\theta $%
.
\end{proposition}

In equilibrium, both on- and off-platform consumers have the same
information as in Section \ref{sec:epl}, where the platform has initially
exclusive information about $\theta $ (henceforth, the \textquotedblleft
baseline\textquotedblright\ model). Thus, any seller who participates in the
managed campaign mechanism offers the optimal menu in Proposition \ref{menu}%
. Facing informed consumers, however, changes the seller's value of turning
down the platform's offer, because consumers who know their values visit
their favorite seller off-platform regardless of the identity of the
sponsored seller on the platform. Suppose a consumer sees a product by a
seller they did not expect on the platform. Under our symmetry refinement,
this consumer continues to believe that all sellers offer symmetric menus
off the platform.

Therefore, each seller $j$ can choose not to participate in the managed
campaign and poach any consumer to whom they offer the highest value: $%
\theta _{j}=\max_{k}\theta _{k}$. Seller $j$ can achieves this by offering
the consumers an off-platform information rent $\widehat{U}_{j}(\theta _{j})$
above the level $\widehat{U}_{j}^{\ast }(\theta _{j})$ in (\ref{U0s}) which
is offered by the competitors in equilibrium. When contemplating such a
menu, the deviating seller $j$ solves the following problem: 
\begin{eqnarray}
\hat{\Pi} &\triangleq &\max_{\widehat{q},\widehat{U}}\int_{\theta
_{L}}^{\theta _{H}}[\theta _{j}\widehat{q}(\theta _{j})-\widehat{q}(\theta
_{j})^{2}/2-\widehat{U}(\theta _{j})]\left[ 
\begin{array}{c}
(1-\lambda )G^{J-1}(\theta _{j})g(\theta _{j}) \\ 
+\lambda F^{J-1}(\theta _{j})f(\theta _{j})%
\end{array}%
\right] \text{d}\theta _{j}  \label{pihat} \\
&&\text{s.t. }\widehat{U}(\theta _{j})\geq \widehat{U}_{j}^{\ast }(\theta
_{j}).  \label{Uhat}
\end{eqnarray}

The equilibrium rent function (\ref{U0s}) in the baseline model satisfies
the constraint (\ref{Uhat}) and yields a strictly larger profit. Therefore,
the sellers' outside option with known values exceeds the outside option $%
\overline{\Pi }$ of the baseline model characterized in (\ref{eq_ou}).

The deviating seller can do even better by offering the optimal menu of
products when consumer values are distributed according to the mixture $%
\left( 1-\lambda \right) G^{J}+\lambda F^{J}$. These quality levels are
given by%
\begin{equation}
\widehat{q}(\theta _{j})=\max \left\{ 0,\theta _{j}-\frac{1-\lambda
F^{J}(\theta _{j})-(1-\lambda )G^{J}(\theta _{j})}{\lambda JF^{J-1}(\theta
_{j})f(\theta _{j})+(1-\lambda )JG^{J-1}(\theta _{j})g(\theta _{j})}\right\}
.  \label{qmixture}
\end{equation}%
The equilibrium quality in (\ref{qmixture}) is larger for every value than
the equilibrium $\widehat{q}_{j}^{\ast }(\theta _{j})$ in (\ref{q0s}) and
yields higher utility to the consumers. Thus, constraint (\ref{Uhat}) does
not bind in the optimal deviation---the best off-platform menu for seller $j$
offers a higher utility level to $j$'s favorite consumers than all other
sellers' menus.\footnote{%
Note that the deviating seller cannot attract any consumer who values a
competitor's products more than their own. This is because those consumers
still face search costs off-platform and would not learn that the deviating
seller has lowered their prices.}

To summarize, Proposition \ref{prop_knowntype} shows that the equilibrium
advertising budgets are qualitatively different when the on-platform
consumers know their values from when they learn their values through the
platform's information. After all, the platform loses the ability to steer
the consumer. In the absence of additional information, the platform cannot
grant monopoly power to any seller by displaying their advertisement and
recommending their products. Without additional information, each consumer
evaluates the different products independently of the recommendation
implicit in the ad. In turn, the value to a seller of showing an
advertisement decreases, as does the willingness to pay for the platform's
services.

To quantify the value of the platform's steering power, fix the value
distribution $F$ and consider the distribution $G$ of expected values
generated by an imperfectly informative signal that each (on- and
off-platform) consumer observes about their value. Denote by $t^{\ast }(G)$
the equilibrium advertising budgets in (\ref{tstar}) under distribution $G$.
Now let the consumer's signal become arbitrarily precise, so that the
distribution $G$ of expected converges to the distribution $F$ of values.
The equilibrium menu for the limit case can be obtained by setting $G=F$ in (%
\ref{q0s}). Proposition \ref{prop_knowntype} then implies the following
observation.

\begin{corollary}[Value of Additional Information]
\label{cor:add}\qquad \newline
For all $J>1$, the platform gains strictly positive profit from any
information advantage:%
\begin{equation*}
\lim_{G\rightarrow F}t^{\ast }(G)>t^{\ast }\left( F\right) \text{.}
\end{equation*}
\end{corollary}

Corollary \ref{cor:add} has important implications for a platform's choice
of information design to which we turn in Section \ref{sec_infodesign}. In
particular, the equilibrium advertising budgets jump up as soon as the
platform has any informational advantage relative to the consumers. This
suggest that some degree of information asymmetry---revealing some
additional information to consumers---is always part of the optimal design.

\subsection{Organic Links\label{sec_compdesign}}

In the equilibrium of our baseline model, the consumer chooses the
advertised seller who offers the highest value, and only considers that
seller's on- and off-platform offers. However, in practice, platforms may
also display ``organic links'' that provide additional, free information to
consumers. We extend our model to consider a scenario where the platform
shows all off-platform prices to consumers. In this setting, on-platform
consumers do not incur a search cost and can buy from any seller without
incurring search costs.

Sellers still advertise the socially efficient product varieties and set
prices to make the showrooming constraints bind. The platform assigns the
sponsored link to each consumer's favorite seller, but the off-platform
menus can now affect market shares on the platform. Sellers can attract some
of their competitors' on-platform consumers by offering lower prices off the
platform. These consumers would not learn about the lower prices without the
presence of organic links.

To calculate the sellers' market shares of on-platform consumers, we
consider the off-platform information rents $\widehat{U}_{j}(\theta _{j})$.
The outside option of the on-platform consumer $\theta $ is given by $%
\max_{j}\{\widehat{U}_{j}(\theta _{j})\}.$ For a symmetric strategy profile
by all sellers $k\neq j$, and for each value $\theta _{j}$, we define the
indifferent value $\theta _{k}^{\ast }(\theta _{j})$ as%
\begin{equation}
\widehat{U}_{k}(\theta _{k}^{\ast }(\theta _{j}))=\widehat{U}_{j}(\theta
_{j})\text{.}  \label{thetaminus}
\end{equation}

With $\theta _{k}^{\ast }=\theta _{k}^{\ast }(\theta _{j})$ defined as in (%
\ref{thetaminus}), seller $j$'s best-response problem is given by%
\begin{align}
\max_{\widehat{q},\widehat{U}}& \left( 1-\lambda \right) \int_{\theta
_{L}}^{\theta _{H}}\underset{\text{off-platform sales}}{\underbrace{[\theta
_{j}\widehat{q}(\theta _{j})-\widehat{q}(\theta _{j})^{2}/2-\widehat{U}%
(\theta _{j})]}}G^{J-1}(\theta _{j})g(\theta _{j})\text{d}\theta _{j}
\label{orgprofit} \\
& +\lambda \int_{\theta _{L}}^{\theta _{H}}\underset{\text{on-platform sales}%
}{\underbrace{(\theta _{j}^{2}/2-\widehat{U}(\theta _{j}))}}\min
\{F^{J-1}(\theta _{k}^{\ast }(\theta _{j})),F^{J-1}(\theta _{j})\}f(\theta
_{j})\text{d}\theta _{j}  \notag \\
& +\lambda \int_{\theta _{L}}^{\theta _{H}}\underset{\text{off-platform sales%
}}{\underbrace{[\theta _{j}\widehat{q}(\theta _{j})-\widehat{q}(\theta
_{j})^{2}/2-\widehat{U}(\theta _{j})]}}\max \{0,F^{J-1}(\theta _{k}^{\ast
}(\theta _{j}))-F^{J-1}(\theta _{j})\}f(\theta _{j})\text{d}\theta _{j}\text{%
.}  \notag
\end{align}

The first term in (\ref{orgprofit}) captures the off-platform consumers. The
second term captures the sales to on-platform consumers for which seller $j$
offers both the highest utility level $\widehat{U}_{j}$ and the highest
marginal value $\theta _{j}$. The third term, whenever positive, captures
on-platform consumers with $\theta _{j}<\max_{k\neq j}\theta _{k}$ to whom
seller $j$ nonetheless offers the highest utility level $\widehat{U}_{j}$.
Seller $j$ is not advertised to these consumers, who instead showroom and
buy from seller $j$ off the platform.

Seller $j$'s problem can therefore be restated as follows. Undercutting the
other sellers (so that $\theta _{k}^{\ast }(\theta _{j})>\theta _{j}$)
yields some additional on-platform consumers who buy off-platform.
Conversely, raising prices above the other sellers' (so that $\theta
_{k}^{\ast }(\theta _{j})<\theta _{j}$) causes seller $j$ to lose some
consumers who would otherwise buy on platform (i.e., a higher quality
product at a higher price, relative to off platform sales). Therefore,
relative to the baseline model with a sponsored link only, each seller $j$
has an incentive to raise $\theta _{k}^{\ast }$ through a higher $\widehat{U}%
_{j}$. This incentive, which is entirely due to organic links, explains the
result in Proposition \ref{prop_organic}, where we compare the game with
organic links to the baseline setting.

\begin{proposition}[Equilibrium with Organic Links]
\label{prop_organic}\strut

\begin{enumerate}
\item The equilibrium quality and utility levels $\widehat{q}_{j}^{\ast
}(\theta _{j})$ and $\widehat{U}_{j}^{\ast }(\theta _{j})$ are weakly higher
for all $\theta _{j}$ with organic links than without.

\item The sellers' profits are lower and their outside options are higher
with organic links than without.
\end{enumerate}
\end{proposition}

To establish this result, we consider the symmetric equilibria of the
subgame following the platform's announcement of an advertising budget $t$.
The symmetric equilibrium quality levels of this game can be characterized
through a system of differential equations, as in \cite{bona11}. We show
that competition among the sellers is fiercer in any such equilibrium. In
particular, quality and utility levels are higher and sellers' profits are
lower than without organic links. Conversely, the sellers' outside options
in any continuation equilibrium are higher than in the baseline model, and
the platform demands a lower advertising budget.

Intuitively, the presence of organic information benefits consumers but
reduces the platform's ability to restrain competition and extract surplus
from sellers. In a symmetric equilibrium, each seller's market segment
consists of all consumers who like the products the most. These market
shares, unlike the baseline case, are endogenous to the choice of $\widehat{U%
}_{j}$. Because the off-platform menus can affect the on-platform market
shares, offering higher rents to consumers has an additional benefit. The
equilibrium utility and quality levels are then higher than without organic
links, the on-path gross profit of the sellers are lower, and consumer
surplus is higher.

However, the sellers' profits net of the advertising budget are equal to the
value of their outside option. With organic links, any seller can respond to
competitors' prices without participating in the mechanism. Let $\theta
_{k}^{\ast }$ be given by (\ref{thetaminus}), with $\widehat{U}_{k}=\widehat{%
U}_{j}^{\ast }$. Because all of the sales necessarily happen off the
platform, a deviating seller $j$ can then obtain profit of%
\begin{equation}
\widetilde{\Pi }_{j}\triangleq \max_{\widehat{q},\widehat{U}}\int_{\theta
_{L}}^{\theta _{H}}\left[ \theta _{j}\widehat{q}(\theta _{j})-\widehat{q}%
(\theta _{j})^{2}/2-\widehat{U}(\theta _{j})\right] \left[ 
\begin{array}{c}
\left( 1-\lambda \right) G^{J-1}(\theta _{j})g(\theta _{j}) \\ 
+\lambda F^{J-1}(\theta _{k}^{\ast }(\theta _{j}))f(\theta _{j})%
\end{array}%
\right] \text{d}\theta _{j}.  \label{orgdev}
\end{equation}

Unlike in the baseline model, each deviating seller has the opportunity to
win over some (but not necessarily all) on-platform consumers for which $%
\theta _{j}\geq \max_{k\neq j}\theta _{k}$. The outside option $\widetilde{%
\Pi }_{j}$ in (\ref{orgdev}) then exceeds the value $\overline{\Pi }_{j}$
defined in (\ref{eq_ou}). In other words, the sellers' outside options are
higher with organic links than without, and the equilibrium advertising
budgets are correspondingly lower.

Finally, recall that with symmetric information and no organic links, the
deviating seller wins \emph{all} on-platform consumers for which $\theta
_{j}\geq \max_{k\neq j}\theta _{k}$. Thus, the outside option $\hat{\Pi}_{j}$
defined in (\ref{pihat}) is even higher than $\widetilde{\Pi }_{j}$ in (\ref%
{orgdev}). Because the equilibrium menus with organic links do not change if
consumers know their values, it is possible that the platform might
profitably raise the requested advertising budget by showing organic links
if consumers are already fully informed about their values.

\subsection{Privacy Protection\label{privacy}}

Up to this point, we have not limited the platform's ability to share
information about the values of the buyers, $\theta $, with the sellers. In
reality, the extent of data sharing may be restricted by regulation or
design choices made by the platform. Now, we assume that the value of the
advertised product $\theta _{j}$ is shared with consumers on the platform,
but not with sellers. Sellers are only allowed to base their offers on the 
\emph{ranking} of consumer valuations $\theta _{j}$ within a \emph{cohort}
of consumers, where each consumer ranks the $J$ sellers in the same way.

Despite this change, efficient matching of sellers and consumers remains
feasible. However, consumers on the platform still have some private
information about their preferences. Unlike the baseline case where each
seller could make personalized offers to consumers, cohort-based ads mean
that seller $j$ only knows the distribution of consumer values based on the
order statistics implied by their cohort. Consequently, each seller must
screen consumers both on and off the platform.

The symmetric equilibrium menus under cohort-based ads are the solution to
two linked screening problems. In the on-platform problem, the showrooming
constraints act as value-dependent participation constraints. To solve this
problem, we strengthen the ranking of the distributions $F$ and $G$ by
assuming that\ the on-platform distribution $F$\ dominates the off-platform
distribution $G$\ in the \emph{likelihood-ratio order, }denoted $F\succ
_{lr}G$. The distribution $F$ dominates $G$ in the likelihood-ratio order if 
$g\left( \theta _{j}\right) /f\left( \theta _{j}\right) $ is decreasing in $%
\theta _{j}$ (see Definition 1.C.1 in \cite{shsh94}). We only require that
the \emph{likelihood-ratio order} is maintained\emph{\ }over the range of
values that receive a positive quality under $F$. We define the (Myerson)
virtual values for the two distributions as: 
\begin{equation*}
\phi _{F}(\theta _{j}):=\theta _{j}-\frac{1-F^{J}(\theta _{j})}{%
JF^{J-1}(\theta _{j})f(\theta _{j})}\text{, and }\phi _{G}(\theta
_{j}):=\theta _{j}-\frac{1-G^{J}(\theta _{j})}{JG^{J-1}(\theta _{j})g(\theta
_{j})}.
\end{equation*}

\begin{proposition}[Cohort Targeting]
\label{floc}\qquad \newline
In the unique symmetric equilibrium, each seller offers quality levels 
\begin{equation*}
\widehat{q}_{j}^{\ast }(\theta _{j})=q_{j}^{\ast }(\theta _{j})=\max \left\{
0,\theta _{j}-\frac{1-\lambda F^{J}(\theta _{j})-(1-\lambda )G^{J}(\theta
_{j})}{\lambda JF^{J-1}(\theta _{j})f(\theta _{j})+(1-\lambda
)JG^{J-1}(\theta _{j})g(\theta _{j})}\right\}
\end{equation*}%
if and only if $F^{J}\succ _{lr}G^{J}$ for all $\theta _{j}$ such that $\min
\{\phi _{F^{J}}(\theta _{j}),\phi _{G^{J}}(\theta _{j})\}\geq 0$.
\end{proposition}

Proposition \ref{floc} shows that if the distribution of the highest $\theta
_{j}$ dominates that of the highest $m_{j}$ in likelihood ratio (over the
relevant range), then each seller offers the same menu to consumers both on
and off the platform. In this menu, the equilibrium quality schedule is the
same as in (\ref{qmixture}), i.e., the \cite{muro78} quality level for a
mixture with weights $\left( \lambda ,1-\lambda \right) $ of the
distributions of the highest order statistics of $\theta $ and $m$,
respectively. Cohort-based ads thus yield higher quality provision
off-platform but lower quality on-platform relative to the baseline model
with full disclosure of the value $\theta $.

A critical implication of Proposition \ref{floc} is that all consumers
receive higher information rents relative to the baseline setting because of
the greater quality provision off the platform. Total surplus can also be
higher as a consequence of greater off-platform quality, although
on-platform quality is lower. Finally, as the sellers' outside options are
unchanged relative to the baseline model, the equilibrium advertising
budgets are unambiguously lower.

\subsection{Information Design\label{sec_infodesign}}

In the analysis so far we have assumed that the platform reveals the true
value $\theta _{j}$ of the sponsored brand $j$ to the consumer. Proposition %
\ref{optimalmech} has shown that the optimal mechanism is then to match each
consumer with their favorite seller $j^{\ast }=\arg \max \theta _{j}$
according to their true preferences, which is the efficient managed campaign
mechanism.

In this section, we investigate the optimal information design by the
platform. We derive conditions under which full information revelation is
approximately or exactly optimal, and conditions for no information
revelation to be optimal. We then focus on the case of uninformed
off-platform consumers and analyze how the optimal information policy
changes with the platform's size $\lambda$ and the distributions of values $%
F(\theta)$. Finally, we discuss the information revealed by the matching
mechanism when the platform does not provide full information about
consumers' value for the sponsored seller.

We assume that the platform knows each consumer's value $\theta $ and their
expected value $m$. This assumption simplifies the analysis, but it is also
a reasonable approximation since if the platform knows every consumer's true
preferences, it may also have information about their past \emph{experiences}%
.\footnote{%
See \cite{lima20} for a formal distinction.} For instance, the platform may
have access to the consumer's cookies and browsing history, which would
enable it to estimate the consumer's expected value.

In a managed campaign, each seller $j$ maximizes total profit by choosing
prices and product qualities given the platform's designed information. The
platform, in turn, maximizes the sellers' profits by choosing the
distribution of expected values. Advertising budgets then extract the
sellers' willingness to pay for this information. Hence, we can think of the
platform as designing both the information and the prices to maximize the
sellers' profits.

However, revealing information to consumers presents a trade-off for the
platform. On the one hand, better information improves the efficiency of
matching between consumers and sellers and product varieties. On the other
hand, better information increases the consumer's expected rent off the
platform, tightens the showrooming constraint, and reduces the sellers'
willingness to pay. To solve the platform's problem, we first show that it
is without loss to focus on symmetric information structures.

\begin{lemma}[Symmetric Information]
\label{tan}\strut

\noindent The optimal information structure enables trade under symmetric
(possibly partial) information on the platform.
\end{lemma}

This result first appeared as Lemma 1 in \cite{bebg22}. The intuition in our
model is that (i) holding off-platform menus fixed, the platform increases
the sellers' profits by eliminating the consumers' private information; and
(ii) any private information signaled by the sellers to the consumers
through their prices can be profitably revealed up-front to the consumers.

\paragraph{Large Platform}

We begin by considering the limit case where $\lambda =1.$ In other words, a
measure one of consumers shop on the platform, and hence sellers have no
reason to post off-platform menus. In this case, we show that the platform
maximizes the advertising budgets by committing to the efficient managed
campaign and by fully revealing $\theta _{j}$ to both consumer $\theta $ and
to the sponsored seller $j$.

\begin{proposition}[Large Platform]
\label{prop_fullinfo}\strut

\noindent When $\lambda =1,$ for any number of sellers $J$ and distributions 
$F$ and $G$, it is optimal for the platform to match consumer $\theta $ to
the efficient seller $j^{\ast }$ and to fully reveal $\theta _{j^{\ast }}$.
\end{proposition}

When the platform becomes arbitrarily large $\left( \lambda \rightarrow
1\right) $, the rents off-platform vanish and the sponsored seller can
appropriate the entire surplus it generates. The sponsored seller's profit
under complete and symmetric information on each value $\theta _{j}$ is then
given by the first-best surplus $\pi _{j}^{\ast }(\theta )=\theta _{j}^{2}/2$%
. Because $\pi _{j}^{\ast }\left( \cdot \right) $ is strictly convex, the
platform-optimal information design reveals to each consumer their true
value for the sponsored seller. Furthermore, by Proposition \ref{optimalmech}%
, it is optimal to match consumers and sellers efficiently when the platform
reveals all the information.

\paragraph{Zero Private Information}

We now characterize the platform's optimal information policy in the special
case where the off-platform consumers have no private information about
their expected values. Thus, the distribution $G$ places a unit mass on the
expected value $\mu \triangleq \mathbb{E}_{F}[\theta _{j}]$ for all $j$,
thus $G(m_{j})=\mathbf{1}_{\left\{ {m_{j}\geq \mu }\right\} }$.

Because the off-platform consumers have no private information, each seller $%
j$ offers just one product of quality $\widehat{q}_{j}\in \mathbb{R}_{+}$
off-platform at a price that extracts the consumer's expected willingness to
pay, i.e., $\widehat{p}_{j}=\mu \widehat{q}_{j}$. As we know from our
baseline model, the seller's choice of off-platform quality $\widehat{q}_{j}$
directly controls the information rent of all on-platform consumers. In
particular, a consumer with an expected value $\theta _{j}$ (given the
information revealed to them by the platform) obtains a rent 
\begin{equation*}
U(\theta _{j},\widehat{q}_{j})=\max \{0,(\theta _{j}-\mu )\widehat{q}_{j}\}
\end{equation*}%
when buying from seller $j$. Consequently, seller $j$'s on-platform profit
as a function of the realized value $\theta _{j}$ are given by 
\begin{equation}
\pi (\theta _{j},\widehat{q}_{j})=\theta _{j}^{2}/2-U(\theta _{j},\widehat{q}%
_{j}).  \label{eq:persprof}
\end{equation}%
Figure \ref{fig_uninf} illustrates the profit function $\pi (\cdot ,\widehat{%
q}_{j})$ for an example where $\mu =1/2$ and $\widehat{q}_{j}=1/2.$ The
seller extracts the entire willingness to pay of all on-platform values $%
\theta _{j}\leq 1/2$ but leaves a rent to values $\theta _{j}\geq 1/2$
(hence the downward kink).

%TCIMACRO{\TeXButton{B}{\begin{figure}[htbp]\centering}}%
%BeginExpansion
\begin{figure}[htbp]\centering%
%EndExpansion
\includegraphics[width=.53\textwidth]{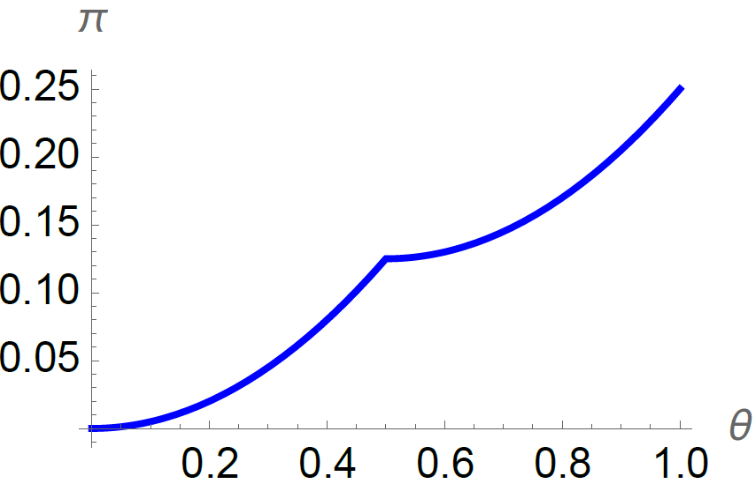}
%TCIMACRO{\TeXButton{caption}{\caption{On-platform profit levels}}}%
%BeginExpansion
\caption{On-platform profit levels}%
%EndExpansion
\label{fig_uninf}%
%TCIMACRO{\TeXButton{E}{\end{figure}}}%
%BeginExpansion
\end{figure}%
%EndExpansion

As a first step towards characterizing the optimal information design, we
establish that the platform shows each consumer an ad by their favorite
seller and reveals an informative signal about their value $\theta _{j^{\ast
}}$.

\begin{lemma}[Efficient Steering]
\strut

\noindent If consumers are uninformed, the efficient managed campaign
matching mechanism is optimal for the platform.
\end{lemma}

To gain intuition, observe that the seller's profit from value $\theta _{j}$
is $\pi (\theta _{j},\widehat{q}_{j})$, which is strictly increasing in $%
\theta _{j}$. Thus, the platform's payoff increases when the distribution of
the underlying \textquotedblleft state\textquotedblright\ (i.e., $\theta
_{j} $) improves in the first order stochastic sense. Because the
distribution of the highest order statistic $F^{J}\left( \theta _{j}\right) $
first-order dominates the distribution of values $\widehat{F}$ that is
induced by any other matching mechanism, the sender (i.e., the platform)
chooses to design information about the consumer's highest value component $%
\theta _{j}$.

Therefore, we can write the platform's problem as 
\begin{align}
& \max_{\widehat{q}_{j},\widehat{F}}\left[ (1-\lambda )(\mu \hat{q}_{j}-%
\widehat{q}_{j}^{2}/2)+\lambda \int_{\theta _{L}}^{\theta _{H}}(\theta
_{j}^{2}/2-\max \{0,(\theta _{j}-\mu )\widehat{q}_{j}\})\text{d}\widehat{F}%
(\theta _{j})\right]  \label{idobj} \\
& \text{s.t. }F^{J}\succ \widehat{F}.  \notag
\end{align}%
In order to solve this problem, we adapt the toolkit of \cite{dwma19} for
persuasion problems where the receiver's posterior mean is a sufficient
statistic for their beliefs. We first fix an off-platform quality level $%
\widehat{q}$ and optimize over information structures. We then characterize
the optimal quality level off platform.

\begin{proposition}[Optimal Information Design]
\label{prop_id}\strut

\noindent Fix $\widehat{q}$ and suppose the off-platform consumers have zero
private information.

\begin{enumerate}
\item There exist two thresholds $x_{1}\leq \mu \leq x_{2}$ such that the
optimal distribution of posteriors $\widehat{F}^{\ast }(\theta _{j})$
coincides with $F^{J}(\theta _{j})$ on $[0,x_{1}]$ and $[x_{2},1]$ and has
an atom at $\mu $.

\item The pair of optimal thresholds $\left( x_{1},x_{2}\right) $ are the
unique solution to%
\begin{eqnarray*}
x_{1}+2\widehat{q} &=&x_{2} \\
\mathbb{E}_{F^{J}}[\theta \mid x_{1}\leq \theta \leq x_{2}] &=&\mu .
\end{eqnarray*}
\end{enumerate}
\end{proposition}

Thus, the platform matches consumers and sellers efficiently but does not
enable efficient trade for all values.\ In particular, values closest to the
mean of the marginal distribution $\mu $ all receive the efficient quality 
\emph{for the average value. }The pooling region allows the seller to
optimally trade off higher on-platform profit with lower off platform rents.
Figure \ref{fig_uninf_sol}\ illustrates the solution for the case of $%
\lambda =3/8$, with $F\left( \theta _{j}\right) =\theta _{j}$ and $J=2.$

%TCIMACRO{\TeXButton{B}{\begin{figure}[htbp]\centering}}%
%BeginExpansion
\begin{figure}[htbp]\centering%
%EndExpansion
\includegraphics[width=.52\textwidth]{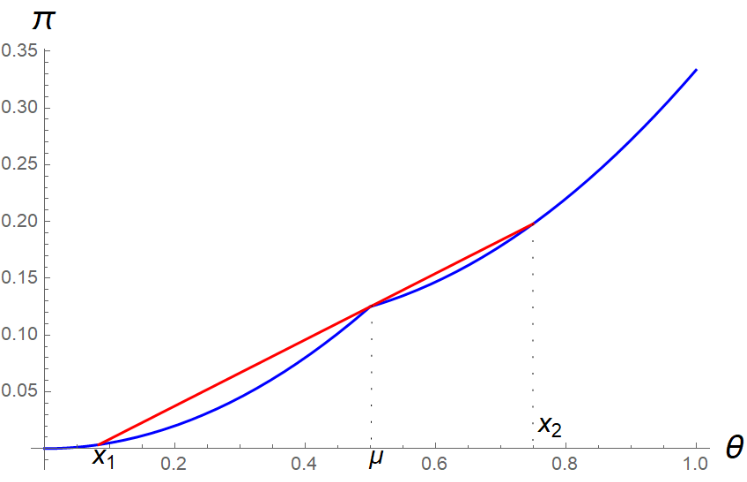}
%TCIMACRO{%
%\TeXButton{caption}{\caption{Optimal Separating and Pooling Intervals}}}%
%BeginExpansion
\caption{Optimal Separating and Pooling Intervals}%
%EndExpansion
\label{fig_uninf_sol}%
%TCIMACRO{\TeXButton{E}{\end{figure}}}%
%BeginExpansion
\end{figure}%
%EndExpansion

Having characterized the optimal information design for any choice of
quality off the platform, we can compute the optimal $\widehat{q}^{\ast }$
from the sponsored seller's profit function.%
\begin{align*}
\Pi (\widehat{q})=& \lambda \int_{0}^{x_{1}(\widehat{q})}(\theta _{j}^{2}/2)%
\text{d}F^{J}(\theta _{j})+\left( \mu ^{2}/2\right) \left( F^{J}(x_{2}(%
\widehat{q}))-F^{J}(x_{1}(\widehat{q}))\right) \\
+& \lambda \int_{x_{2}(\widehat{q})}^{1}\left( \theta _{j}^{2}/2-\widehat{q}%
(\theta _{j}-\mu )\right) \text{d}F^{J}(\theta _{j})+\left( 1-\lambda
\right) \left( \mu \widehat{q}-\widehat{q}^{2}/2\right) .
\end{align*}

It is then immediate to show that as $\lambda \rightarrow 1$ (as in
Proposition \ref{prop_fullinfo}), the optimal $\widehat{q}^{\ast
}\rightarrow 0$ and $x_{1},x_{2}\rightarrow \mu $ so that the platform
reveals the consumer's value with probability one. In some special cases,
the solution is in closed form and yields the conclusion of Proposition \ref%
{prop_fullinfo} even if $\lambda $ is bounded away from $1$.

\paragraph{Discussion}

When $\lambda \in \left( 0,1\right) $ and the distribution $G$ of expected
values $m_{j}$ is not degenerate, the problem becomes significantly more
complex. If the platform does not know the consumer's expected value, then
it faces a persuasion problem where the receiver's private value is
correlated with the state, unlike in \cite{komy17}. A potentially fruitful
approach to this problem could be to focus attention to the case of \textit{%
public persuasion}, i.e., to signal structures that do not condition on the
consumer's expected value $m$.

\section{Platform Size and Competition\label{cs_lambdasec}}

Having examined the informational sources of the platform's bargaining
power, we now return to our baseline setting of Section \ref{sec:search} to
study the role of the size of the platform and the competition among the
sellers. We first investigate how the market share of the platform $\lambda $
affects the welfare and distribution of the social surplus. We then analyze
how an increase in competition in terms of the number of competing sellers
affects the welfare outcomes on and off the platform.

\paragraph{Platform Size}

The opportunity cost of serving consumers off the platform increases as the
platform becomes (exogenously) larger. Intuitively, the information rents of
the off-platform consumers must also be paid to a mass $\lambda $ of
on-platform consumers. This should lead to further distortions in the
off-platform quality levels. We formalize this intuition in Proposition \ref%
{compstat}.

\begin{proposition}[Platform Size]
\label{compstat}\strut \vspace*{-0.1in}

\begin{enumerate}
\item The equilibrium quality levels $\widehat{q}_{j}^{\ast }(\theta _{j})$
are decreasing in $\lambda $ for all $\theta _{j}<\theta _{H}$, and the
information rents $\widehat{U}_{j}^{\ast }(\theta _{j})$ are decreasing in $%
\lambda $ for all $\theta _{j}$.

\item For every $\theta _{j}<\theta _{H}$, there exists $\bar{\lambda}<1$
such that $\widehat{q}_{j}^{\ast }(\theta _{j})=0$ for all $\lambda \geq 
\bar{\lambda}.$
\end{enumerate}
\end{proposition}

In Figure \ref{lcs}, we illustrate how the off-platform quality provision
changes as the market share $\lambda $ of the platform increases.

%TCIMACRO{\TeXButton{B}{\begin{figure}[htbp]\centering}}%
%BeginExpansion
\begin{figure}[htbp]\centering%
%EndExpansion
\includegraphics[width=.90\textwidth]{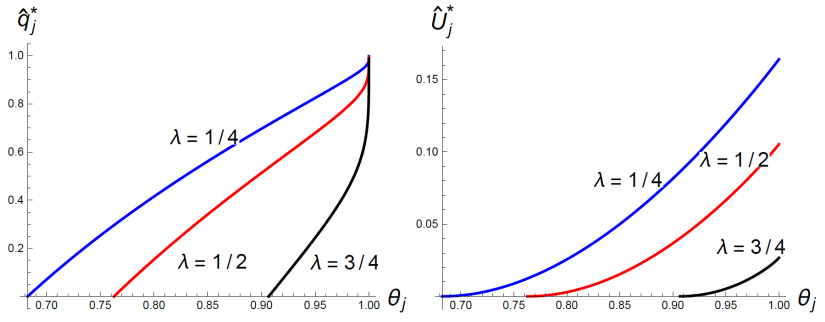}
%TCIMACRO{%
%\TeXButton{caption}{\caption{Off-Platform Menus, $J=3,G(m_j)=m_j,F(\theta_j)=\textrm{Beta}(\theta_j,1/4,1/4)$.}}}%
%BeginExpansion
\caption{Off-Platform Menus, $J=3,G(m_j)=m_j,F(\theta_j)=\textrm{Beta}(\theta_j,1/4,1/4)$.}%
%EndExpansion
\label{lcs}%
%TCIMACRO{\TeXButton{E}{\end{figure}}}%
%BeginExpansion
\end{figure}%
%EndExpansion

The on-platform allocation remains unchanged and is given by the socially
efficient quality provision. However, as the platform grows larger, each
seller attempts to minimize the information rents on the platform and in
turn renders the menu off the platform less attractive. Thus, for every
value $\theta _{j}$, the equilibrium quality-match off the platform $%
\widehat{q}_{j}^{\ast }(\theta _{j})$\ decreases, the price per unit of
quality increases, and the consumer surplus $\widehat{U}_{j}^{\ast }$\ off
the platform decreases as the size of the platform increases.\footnote{%
Recent work by \cite{vale22} considers many-to-many matching with congestion
effects on each side in the \cite{gopa16} framework. In our model,
congestion effects arise endogenously on the consumers' side, because the
more consumers visit the platform the fewer options are available
off-platform, which drives information rents down.}

\paragraph{Number of Sellers}

As the number of sellers increases, a larger number $J$\ of draws for each
value $\theta _{j}$ improves the distributions $F^{J}$ and $G^{J}$ in the
likelihood-ratio order. This leads to lower information rents. In the limit,
a seller will know that every consumer who shops on their site, or receives
their ads, has a value near $\theta _{H}$ with probability close to $1$, and
therefore information rents vanish. This result is a direct implication of
the \cite{diamond71} model adapted to our setting. We illustrate this result
for the benchmark case of an off-platform market only (i.e., $\lambda =0$)
in Figure \ref{rcs}.

%TCIMACRO{\TeXButton{B}{\begin{figure}[htbp]\centering}}%
%BeginExpansion
\begin{figure}[htbp]\centering%
%EndExpansion
\includegraphics[width=.90\textwidth]{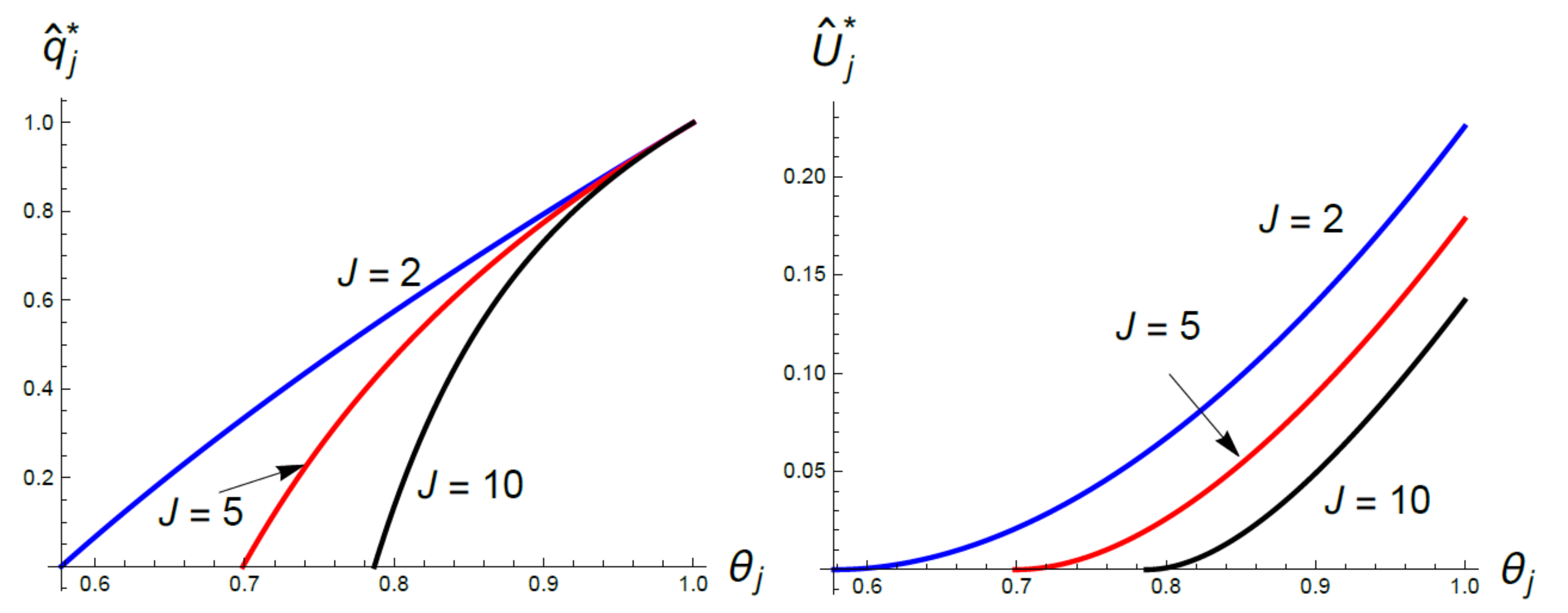}
%TCIMACRO{%
%\TeXButton{caption}{\caption{Off-Platform Menus, $\lambda=0,G(m_j)=m_j,F(\theta_j)=\textrm{Beta}(\theta_j,1/4,1/4)$.}}}%
%BeginExpansion
\caption{Off-Platform Menus, $\lambda=0,G(m_j)=m_j,F(\theta_j)=\textrm{Beta}(\theta_j,1/4,1/4)$.}%
%EndExpansion
\label{rcs}%
%TCIMACRO{\TeXButton{E}{\end{figure}}}%
%BeginExpansion
\end{figure}%
%EndExpansion

Relative to the \cite{diamond71} model, quality distortions decrease faster
in our setting for lower values and slower for higher values. This effect is
due to the interaction of showrooming and the different distributions of
values. In particular, one can show that the additional distortion term in
the equilibrium quality (\ref{q0ss}), i.e.,%
\begin{equation*}
\frac{\lambda }{1-\lambda }\frac{1-F^{J}(\theta _{j})}{JG^{J-1}(\theta
_{j})g(\theta _{j})}
\end{equation*}%
is decreasing in $J$ when $\theta _{j}$ is close to $\theta _{H}$. Thus, for
a small number of sellers, high values of $\theta _{j}$\ receive a higher
quality as $J$ increases while low values of $\theta _{j}$\ receive a lower
quality.

As Figure \ref{jcs} illustrates, this effect may not be sufficient to
generate a larger rent for any value. Furthermore, Proposition \ref{nJ}
shows that as $J$ grows large, every value's quality allocation eventually
decreases in the number of sellers.

%TCIMACRO{\TeXButton{B}{\begin{figure}[htbp]\centering}}%
%BeginExpansion
\begin{figure}[htbp]\centering%
%EndExpansion
\includegraphics[width=.90\textwidth]{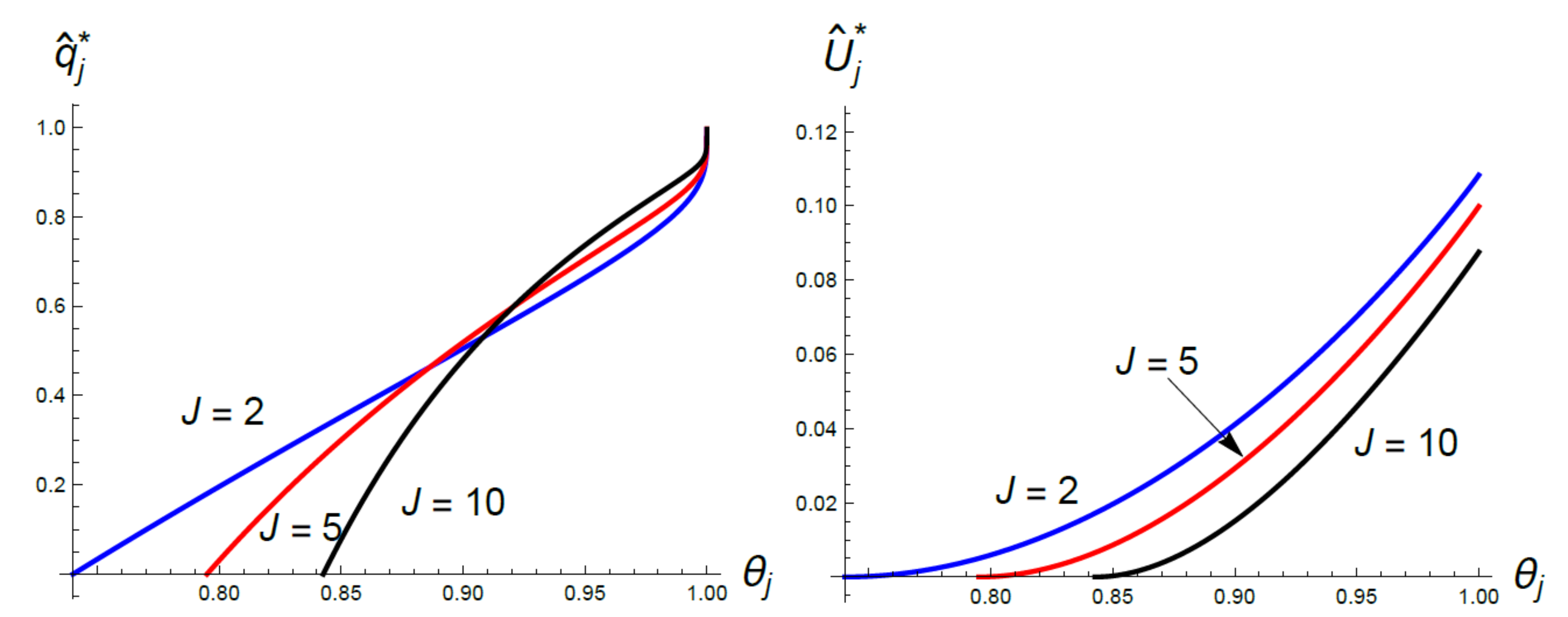}
%TCIMACRO{%
%\TeXButton{caption}{\caption{Off-Platform Menus, $\lambda=1/2,G(m_j)=m_j,F(\theta_j)=\textrm{Beta}(\theta_,1/4,1/4)$.}}}%
%BeginExpansion
\caption{Off-Platform Menus, $\lambda=1/2,G(m_j)=m_j,F(\theta_j)=\textrm{Beta}(\theta_,1/4,1/4)$.}%
%EndExpansion
\label{jcs}%
%TCIMACRO{\TeXButton{E}{\end{figure}}}%
%BeginExpansion
\end{figure}%
%EndExpansion

\begin{proposition}[Number of Sellers]
\label{nJ}$\strut $\vspace*{-0.1in}

\begin{enumerate}
\item For every $\theta _{j}<\theta _{H}$, the equilibrium quality $\widehat{%
q}_{j}^{\ast }(\theta _{j})$ and information rent $\widehat{U}_{j}^{\ast
}(\theta _{j})$ are decreasing in $J$ if $J$ is large enough.

\item For every $\theta _{j}<\theta _{H}$, there exists $\widehat{J}$ such
that $\widehat{q}_{j}^{\ast }(\theta _{j})=0$ for all $J\geq \widehat{J}$.
\end{enumerate}
\end{proposition}

We can then examine the impact of the size of the platform $\lambda $ and of
the number of sellers $J$ on all parties' surplus levels. An immediate
consequence of Propositions \ref{compstat} and \ref{nJ} is that expected
consumer surplus on-platform and off-platform is always decreasing in $%
\lambda $, and is eventually decreasing in $J$ too. At the same time, the
platform's revenue is increasing in both $\lambda $ and $J.$ Furthermore, as 
$J$ grows without bound, the platform captures the entire (first best)\
social surplus it creates. Intuitively, the consumers have no information
rents (as the highest value component is converging in probability to 1),
and therefore sellers need not distort the off-platform menus when they
participate in the platform's mechanism. Figure \ref{fig_alls} illustrates
the platform revenue, consumer surplus, seller surplus (i.e., the outside
option $\overline{\Pi }_{j}$), and the profit generated on the platform for
various numbers of sellers.

%TCIMACRO{\TeXButton{B}{\begin{figure}[htbp]\centering}}%
%BeginExpansion
\begin{figure}[htbp]\centering%
%EndExpansion
\includegraphics[width=.6\textwidth]{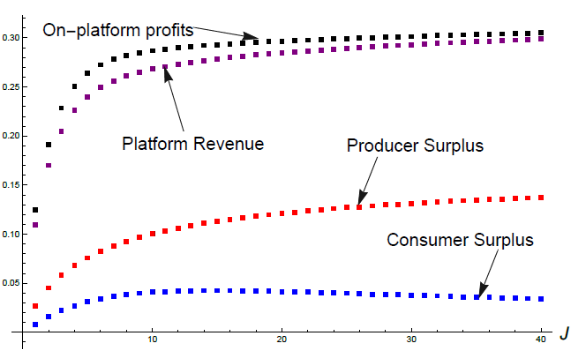}
%TCIMACRO{%
%\TeXButton{caption}{\caption{Surplus Levels, $\lambda=2/3,G(m_j)=m_j,F(\theta_j)=\textrm{Beta}(\theta_j,1/3,1/3)$.}}}%
%BeginExpansion
\caption{Surplus Levels, $\lambda=2/3,G(m_j)=m_j,F(\theta_j)=\textrm{Beta}(\theta_j,1/3,1/3)$.}%
%EndExpansion
\label{fig_alls}%
%TCIMACRO{\TeXButton{E}{\end{figure}}}%
%BeginExpansion
\end{figure}%
%EndExpansion

\section{Conclusion}

We have developed a model that considers competition in the digital economy,
taking into account heterogeneous consumer preferences and products. A
digital platform serves as an intermediary between buyers and sellers and
utilizes its superior information to form matches between them, generating
revenue through digital advertising campaigns. The platform monetizes its
advantage by presenting consumers with their preferred products and earns
revenue by selling access to their attention. However, the ability for
sellers to showcase their products off the platform limits their ability to
price discriminate on the platform and their willingness to pay for
advertising, leading to higher prices on both sales channels as the
platform's user base grows.

Our model is simplified, but it can be extended to consider differentiated
products with varying on- and off-platform presences. This may introduce
distortions in the managed-campaign allocation of advertising space, as
smaller sellers exploit higher margins and consumer search costs.\footnote{%
The recent evidence in \cite{muaa22} is consistent with a mechanism like the
one we outlined.} Our model also assumes \textquotedblleft perfect
steering,\textquotedblright\ but it can be expanded to incorporate multiple
related advertisements to each consumer. Overall, our paper highlights the
role of data in shaping competition and allocating surplus in the digital
economy.

\newpage

\section*{Appendix}

This Appendix contains the proofs of all our results.\bigskip

\begin{proof}[Proof of Proposition \protect\ref{prop_single}]
To derive the low type's optimal quality, we substitute the expression for
the binding incentive compatibility constraint for the high type (\ref{ICH})
into both the on-platform and off-platform profit in objective (\ref{ss_prob}%
). Differentiating with respect to $\widehat{q}\left( \theta _{L}\right) $
yields the result in (\ref{qlow}).\bigskip
\end{proof}

\begin{proof}[Proof of Proposition \protect\ref{prop_consider}]
In any symmetric equilibrium, each consumer $\theta $ with expected value $m$
learns their true value $\theta _{j^{\ast }}$ for the advertised seller $%
j^{\ast }$ and believes that $j^{\ast }=\arg \max_{j}\theta _{j}$ with
probability one. Moreover, each consumer expects identical menus to be
posted off-platform and knows that the rent function $\widehat{U}_{j}(\theta
_{j})$ is strictly increasing in $\theta _{j}$ for all $j$. Therefore, the
consumer searches for the advertised seller's off-platform prices. She does
not search any further and does not learn any other seller's prices. Because
the menu off the platform is incentive compatible, it is sufficient for
consumer $\theta $ to compare the two items $q_{j^{\ast }}(\theta _{j^{\ast
}})$ and $\widehat{q}_{j^{\ast }}(\theta _{j^{\ast }})$.

This holds both on and off the equilibrium path. Indeed, suppose a seller $%
\hat{\jmath}$ deviates and does not participate in the platform's mechanism.
In this case, all consumer with $\hat{\jmath}=\arg \max_{j}\theta _{j}$ are
shown an advertisement by a different seller. These consumers are unable to
detect seller $\hat{\jmath}$'s deviation,.and hence search for the sponsored
seller's menu only.\bigskip
\end{proof}

\begin{proof}[Proof of Proposition \protect\ref{menu}]
Seller $j$'s gross profit (i.e., after paying the platform's required
advertising budget) can be written as%
\begin{align}
\max_{\widehat{q},\widehat{U}}\,& (1-\lambda )\int_{\theta _{L}}^{\theta
_{H}}[\theta _{j}\widehat{q}(\theta _{j})-\widehat{q}(\theta _{j})^{2}/2-%
\widehat{U}(\theta _{j})]G^{J-1}(\theta _{j})\text{d}G(\theta _{j})
\label{eq_prob} \\
& +\lambda \int_{\theta _{L}}^{\theta _{H}}[\theta _{j}^{2}/2-\widehat{U}%
(\theta _{j})]F^{J-1}(\theta )\text{d}F(\theta ),  \notag \\
\text{s.t.}& \text{ }\widehat{U}^{\prime }(\theta _{j})=\widehat{q}(\theta
_{j}),\text{ and }\widehat{U}(\theta _{j})\geq 0\text{ for all }\theta
_{j}\in \left[ \theta _{L},\theta _{H}\right] \text{.}  \notag
\end{align}%
The necessary pointwise conditions for $\widehat{q}$ and $\widehat{U}$ can
be obtained from the control problem with the associated Hamiltonian with
costate variable $\hat{\gamma}(\theta _{j})$:%
\begin{eqnarray*}
H(\theta _{j},\widehat{q},\widehat{U},\hat{\gamma}) &=&(1-\lambda )\left(
\theta _{j}\widehat{q}(\theta _{j})-\widehat{q}(\theta _{j})^{2}/2-\widehat{U%
}(\theta _{j})\right) G^{J-1}(\theta _{j})g(\theta _{j}) \\
&&+\lambda \left( \theta _{j}^{2}/2-\widehat{U}(\theta _{j})\right)
F^{J-1}(\theta _{j})f(\theta _{j})+\hat{\gamma}(\theta _{j})\widehat{q}%
(\theta _{j}).
\end{eqnarray*}%
At a symmetric equilibrium, the optimality conditions are given by 
\begin{eqnarray}
\left( 1-\lambda \right) \left( \theta _{j}-\widehat{q}(\theta _{j})\right)
G^{J-1}(\theta _{j})g(\theta _{j})+\hat{\gamma}(\theta _{j}) &=&0,
\label{eq:qbase} \\
-\left( 1-\lambda \right) G^{J-1}(\theta _{j})g(\theta _{j})-\lambda
F^{J-1}(\theta _{j})f(\theta _{j})+\hat{\gamma}^{\prime }(\theta _{j}) &=&0,
\notag \\
\hat{\gamma}\left( 1\right) &=&0.  \notag
\end{eqnarray}%
Integrating, we obtain%
\begin{equation}
\hat{\gamma}(\theta _{j})=\frac{1}{J}\left( \left( 1-\lambda \right)
G^{J}(\theta _{j})+\lambda F^{J}(\theta _{j})-1\right) .  \label{gammasp}
\end{equation}%
Therefore, the equilibrium quality level is given by%
\begin{equation}
\widehat{q}_{j}^{\ast }(\theta _{j})=\theta _{j}-\frac{1-\left( 1-\lambda
\right) G^{J}(\theta _{j})-\lambda F^{J}(\theta _{j})}{\left( 1-\lambda
\right) JG^{J-1}(\theta _{j})g(\theta _{j})}  \label{q0sproof}
\end{equation}%
if the right-hand side is nonnegative, and nil otherwise, as in (\ref{q0s}).
This ends the proof.\bigskip
\end{proof}

\begin{proof}[Proof of Proposition \protect\ref{optimalmech}]
We argue the optimality of the managed-campaign mechanism in two steps.
First, by Proposition \ref{prop_consider}, any seller $j$ that does not
participate in the mechanism does not make any sales to the consumers on the
platform. This is because every consumer will see an ad by a different
seller $k\neq j$ and will only consider on- and off-platform offers by
seller $k$. Therefore, every seller's outside option consists of offering
the optimal \cite{muro78} menu to their off-platform consumers. This yields
the profit level $\overline{\Pi }_{j}$ in (\ref{eq_ou}), which is a fixed
outside option independent of the menus posted by the other sellers.

Second, consider the coalition of all sellers and the platform. The
coalition's profits are maximized by matching each on-platform consumer $%
\theta $ to seller $j^{\ast }=\arg \max_{j}\theta _{j}$, by matching each
off-platform consumer $m$ to $\hat{\jmath}=\arg \max_{j}m_{j}$, and by
maximizing the seller's profit with respect to the on-platform offers $%
\left( q,U\right) $ and the off-platform menus $(\widehat{q},\widehat{U})$.

The solution to the coalition profit-maximization problem coincides with the
equilibrium outcome of the managed campaign mechanism. In this mechanism,
each seller maximizes profit by offering the socially efficient quality $%
q_{j}\left( \theta \right) =\theta _{j}$ to each on-platform consumer $%
\theta .$ As a result, the platform assigns each seller to the consumers
that value their products the most. Thus, the sellers-platform coalition's
profits are maximized by the equilibrium menus that solve problem (\ref%
{eq_prob}). Because the platform extracts the entire seller surplus in
excess of the fixed outside option (\ref{eq_ou}), no mechanism generates
greater revenue for the platform.\bigskip
\end{proof}

\newpage

\begin{proof}[Proof of Proposition\ \protect\ref{prop_knowntype}]
Suppose on-platform consumers know their value $\theta $. If seller $j$
participates in the mechanism but offers an out-of-equilibrium menu off the
platform, only consumers who search for seller $j$ in equilibrium observe
this deviation. Therefore, every seller that participates in the mechanism
can do no better than to advertise the efficient quality levels and post the
off-platform menus that solve (\ref{eq_prob}). However, if seller $j$ does
not participate, it can match the competitors' information rents $\widehat{U}%
_{k\neq j}$ and attract all the consumers on the platform who value their
products the most. (Under the symmetric beliefs refinement, these consumers
search for seller $j$'s off-platform offer regardless of the ads shown to
them by the platform.) Furthermore, if a nonparticipating seller $j$
maximizes profit with respect to $(\widehat{q},\widehat{U})$ over the
combined off- and on-platform market segments, it solves the problem in (\ref%
{pihat}). This is a standard second-degree price discrimination problem
where consumer values are distributed according to $\lambda F^{J}+(1-\lambda
)G^{J}.$ The optimal quality provision in such a deviation is given by 
\begin{equation}
\widehat{q}(\theta _{j})=\max \bigg\{0,\theta _{j}-\frac{1-\left( 1-\lambda
\right) G^{J}(\theta _{j})-\lambda F^{J}(\theta _{j})}{\lambda
F^{J-1}(\theta _{j})f(\theta _{j})+\left( 1-\lambda \right) G^{J-1}(\theta
_{j})g(\theta _{j})}\bigg\}.  \label{q0mixture}
\end{equation}

Because the quality level $\widehat{q}$ in (\ref{q0mixture}) is pointwise
larger than $\widehat{q}_{j}^{\ast }$ in (\ref{q0sproof}), the resulting
information rent is correspondingly higher for each $\theta _{j}$. Thus, the
deviating seller's optimal choice of menu yields an outside option $\hat{\Pi}%
_{j}$ larger than $\overline{\Pi }_{j}$ in (\ref{eq_ou}). In addition,
because the on-path profit are unchanged relative to the case of
asymmetrically informed consumers, the advertising budget requested by the
platform must decrease.\bigskip
\end{proof}

\begin{proof}[Proof of Proposition \protect\ref{prop_organic}]
We characterize all symmetric equilibria with full participation of the
subgame following the platform's announcement of the required advertising
budget. We first rewrite problem (\ref{orgprofit}) as follows: 
\begin{align*}
\max_{\widehat{q},\widehat{U}}& \left( 1-\lambda \right) \int_{\theta
_{L}}^{\theta _{H}}[\theta _{j}\widehat{q}(\theta _{j})-\widehat{q}(\theta
_{j})^{2}/2-\widehat{U}(\theta _{j})]G^{J-1}(\theta _{j})g(\theta _{j})\text{%
d}\theta _{j} \\
& +\lambda \int_{\theta _{L}}^{\theta _{H}}[\theta _{j}\widehat{q}(\theta
_{j})-\widehat{q}(\theta _{j})^{2}/2-\widehat{U}(\theta _{j})]F^{J-1}(\theta
_{k}^{\ast }(\theta _{j}))f(\theta _{j})\text{d}\theta _{j} \\
& +\lambda \int_{\theta _{L}}^{\theta _{H}}(\theta _{j}^{2}/2-\theta _{j}%
\widehat{q}(\theta _{j})+\widehat{q}(\theta _{j})^{2}/2)\min
\{F^{J-1}(\theta _{k}^{\ast }(\theta _{j})),F^{J-1}(\theta _{j})\}f(\theta
_{j})\text{d}\theta _{j},
\end{align*}%
where, as in (\ref{thetaminus}), $\theta _{k}^{\ast }$ satisfies%
\begin{equation*}
\widehat{U}_{k}\left( \theta _{k}^{\ast }(\theta _{j})\right) =\widehat{U}%
_{j}\left( \theta _{j}\right) .
\end{equation*}%
The associated Hamiltonian can be written as%
\begin{eqnarray*}
H(\theta _{j},\widehat{q},\widehat{U},\hat{\gamma}) &=&\left( 1-\lambda
\right) \left( \theta _{j}\widehat{q}(\theta _{j})-\widehat{q}(\theta
_{j})^{2}-\widehat{U}(\theta _{j})\right) G^{J-1}(\theta _{j})g(\theta _{j})
\\
&&+\lambda \left( \theta _{j}\widehat{q}(\theta _{j})-\widehat{q}(\theta
_{j})^{2}/2-\widehat{U}(\theta _{j})\right) F^{J-1}(\theta _{k}^{\ast
}(\theta _{j}))f(\theta _{j}) \\
&&+\lambda (\theta _{j}^{2}/2-\theta _{j}\widehat{q}(\theta _{j})+\widehat{q}%
(\theta _{j})^{2}/2)\min \{F^{J-1}(\theta _{k}^{\ast }(\theta
_{j})),F^{J-1}(\theta _{j})\}f(\theta _{j}) \\
&&+\gamma (\theta _{j})\widehat{q}(\theta _{j}).
\end{eqnarray*}%
Totally differentiating (\ref{thetaminus}), we obtain 
\begin{equation*}
\frac{dF^{J-1}\left( \theta _{k}^{\ast }\left( \theta _{j}\right) \right) }{d%
\hat{U}_{j}\left( \theta _{j}\right) }=\frac{\left( J-1\right)
F^{J-2}(\theta _{j})f(\theta _{j})}{\widehat{q}(\theta _{j})}\geq 0.
\end{equation*}%
Because seller $j$'s market share is increasing in $\widehat{U}_{j}$ and $%
\theta _{j}^{2}/2\geq \theta _{j}\widehat{q}(\theta _{j})-\widehat{q}(\theta
_{j})^{2}/2$, the Hamiltonian $H$ has a downward kink in $\widehat{U}_{j}$
at $\theta _{k}^{\ast }=\theta _{j}.$ Therefore, every symmetric symmetric
equilibrium satisfies the following necessary conditions,%
\begin{eqnarray}
\left( 1-\lambda \right) \left( \theta _{j}-\widehat{q}(\theta _{j})\right)
G^{J-1}(\theta _{j})g(\theta _{j})+\gamma (\theta _{j})=0\text{,} &&
\label{eq_qorg} \\
-\lambda \frac{\left( J-1\right) F^{J-2}(\theta _{j})f^{2}(\theta _{j})}{%
\widehat{q}(\theta _{j})}\left( \alpha \theta _{j}^{2}/2+\left( 1-\alpha
\right) \left( \theta _{j}\widehat{q}(\theta _{j})-\widehat{q}(\theta
_{j})^{2}/2\right) -\widehat{U}(\theta _{j})\right) &&  \notag \\
+\left( 1-\lambda \right) G^{J-1}(\theta _{j})g(\theta _{j})+\lambda
F^{J-1}(\theta _{j})f(\theta _{j})=\gamma ^{\prime }(\theta _{j})\text{,} &&
\label{eq_costate}
\end{eqnarray}%
for some $\alpha \in \left[ 0,1\right] $.

Now recall the costate equation (\ref{eq:qbase}) in the baseline model.
Because the last two terms on the left-hand side of (\ref{eq_costate}) are
nonnegative, we have 
\begin{equation*}
\gamma ^{\prime }(\theta _{j})\leq \hat{\gamma}^{\prime }(\theta _{j})\text{
for all }\theta _{j}\text{.}
\end{equation*}%
Furthermore, the transversality conditions in the two problems require 
\begin{equation*}
\gamma \left( 1\right) =\hat{\gamma}\left( 1\right) =0.
\end{equation*}%
We can then conclude that%
\begin{equation*}
\hat{\gamma}(\theta _{j})\leq \gamma (\theta _{j})\text{ for all }\theta _{j}%
\text{.}
\end{equation*}%
Together with (\ref{eq:qbase}) and (\ref{eq_qorg}), this implies that the
quality and utility levels $\widehat{q}_{j}(\theta _{j})$ and $\widehat{U}%
_{j}(\theta _{j})$ are weakly higher for all $\theta _{j}$ \emph{with}
organic links than without.

We now show that the sellers' profits are lower and their outside options
are higher with organic links than without. In a symmetric equilibrium with
off-platform quality and rent functions $(\widehat{q},\widehat{U})$, the
seller's profits are given by%
\begin{eqnarray}
\Pi _{j}(\widehat{q},\widehat{U}) &=&\left( 1-\lambda \right) \int_{\theta
_{L}}^{\theta _{H}}\left[ \theta _{j}\widehat{q}(\theta _{j})-\widehat{q}%
(\theta _{j})^{2}/2-\widehat{U}(\theta _{j})\right] G^{J-1}(\theta _{j})%
\text{d}G(\theta _{j})  \label{pi_onpath} \\
&&+\lambda \int_{\theta _{L}}^{\theta _{H}}(\theta _{j}^{2}/2-\widehat{U}%
(\theta _{j}))F^{J-1}(\theta _{j})\text{d}F(\theta _{j}).  \notag
\end{eqnarray}%
The equilibrium menu in the baseline model $(\widehat{q}_{j}^{\ast },%
\widehat{U}_{j}^{\ast })$ maximizes (\ref{pi_onpath}), while the equilibrium
menu maximizes (\ref{orgprofit}) and hence achieves a weakly lower profit
level. Now consider the deviating seller's profit. For any choice of $(%
\widehat{q},\widehat{U})$ off path, the deviation profit are weakly larger
with organic links than without. Without organic links, the deviating seller
posts the \cite{muro78} menu and makes no sales on the platform. With
organic links and posting the same menu, the seller wins a fraction $%
F^{J-1}(\theta _{k}^{\ast }(\theta _{j}))\in \left[ 0,1\right] $ of
on-platform values $\theta _{j}$. Consequently, we have $\widetilde{\Pi }%
_{j}\geq \overline{\Pi }_{j}$, which implies \emph{a fortiori} that the
advertising budgets are lower.\bigskip
\end{proof}

\begin{proof}[Proof of Proposition \protect\ref{floc}]
We construct an equilibrium where each seller $j$ sets on- and off-platform
menus to maximize profit, given that seller $j$\ expects to face all
consumers that rank $j$ the highest. Therefore, consider the joint
optimization problem over menus $\left( q,U\right) $ and $(\widehat{q},%
\widehat{U})$ when facing distributions $F^{J}$ and $G^{J},$ respectively,
under the showrooming constraint. Seller $j$ solves:%
\begin{eqnarray}
&&\max_{q,\widehat{q},U,\widehat{U}}\left[ 
\begin{array}{c}
\lambda \int_{\theta _{L}}^{\theta _{H}}\left( \theta _{j}q(\theta
_{j})-q(\theta _{j})^{2}/2-U(\theta _{j})\right) F^{J-1}(\theta _{j})\text{d}%
F(\theta _{j}) \\ 
+\left( 1-\lambda \right) \int_{\theta _{L}}^{\theta _{H}}\left( \theta _{j}%
\widehat{q}_{j}(\theta _{j})-\widehat{q}(\theta _{j})^{2}/2-\widehat{U}%
(\theta _{j})\right) G^{J-1}(\theta _{j})\text{d}G(\theta _{j})%
\end{array}%
\right]  \label{prob_floc} \\
&\text{s.t.}&\text{~}U^{\prime }(\theta _{j})=q(\theta _{j})  \notag \\
&&\widehat{U}^{\prime }(\theta _{j})=\widehat{q}(\theta _{j})  \notag \\
&&U(\theta _{j})\geq \widehat{U}(\theta _{j})\geq 0.  \notag
\end{eqnarray}%
We now show that the solution to (\ref{prob_floc}) is given by%
\begin{equation}
q_{j}^{\ast }(\theta _{j})=\widehat{q}_{j}^{\ast }(\theta _{j})=\theta _{j}-%
\frac{1-\left( 1-\lambda \right) G^{J}(\theta _{j})-\lambda F^{J}(\theta
_{j})}{J\lambda F^{J-1}(\theta _{j})f(\theta _{j})+J\left( 1-\lambda \right)
G^{J-1}(\theta _{j})g(\theta _{j})}  \label{qbar}
\end{equation}%
if and only if the $F^{J}$ likelihood-ratio dominates $G^{J}$. To this end,
consider the necessary conditions for optimality. These conditions are
sufficient because the problem is linear in $q$, concave in $U$, and
additively separable in these two variables. In particular, the Hamiltonian
is given by%
\begin{eqnarray*}
H &=&\lambda \left( \theta _{j}q(\theta _{j})-q(\theta _{j})^{2}/2-U(\theta
_{j})\right) F^{J-1}(\theta _{j})f(\theta _{j}) \\
&&+\left( 1-\lambda \right) \left( \theta _{j}\widehat{q}(\theta _{j})-%
\widehat{q}(\theta _{j})^{2}/2-\widehat{U}(\theta _{j})\right)
G^{J-1}(\theta _{j})g(\theta _{j}) \\
&&+\gamma (\theta _{j})q(\theta _{j})+\hat{\gamma}(\theta _{j})\widehat{q}%
(\theta _{j})+\bar{\gamma}(U(\theta _{j})-\widehat{U}(\theta _{j})).
\end{eqnarray*}%
The pointwise necessary conditions for this problem are the following:%
\footnote{%
The last condition $\left( \bar{\gamma}\geq 0\right) $ is analogous the one
in \cite{jullien}, Theorem 2. There, the shadow cost of the value-dependent
participation constraint is a cumulative distribution function, i.e., it is
nondecreasing. The multiplier $\bar{\gamma}\left( \theta _{j}\right) $ in
our formulation can be interpreted as the corresponding density function.}%
\begin{eqnarray*}
\left( \theta _{j}-q(\theta _{j})\right) \lambda F^{J-1}(\theta
_{j})f(\theta _{j})+\gamma (\theta _{j}) &=&0 \\
\left( \theta _{j}-\widehat{q}(\theta _{j})\right) \left( 1-\lambda \right)
G^{J-1}(\theta _{j})g(\theta _{j})+\hat{\gamma}(\theta _{j}) &=&0 \\
-\lambda F^{J-1}(\theta _{j})f(\theta _{j})+\bar{\gamma}(\theta _{j})+\gamma
^{\prime }(\theta _{j}) &=&0 \\
-\left( 1-\lambda \right) G^{J-1}(\theta _{j})g(\theta _{j})-\bar{\gamma}%
(\theta _{j})+\hat{\gamma}^{\prime }(\theta _{j}) &=&0 \\
\bar{\gamma}(\theta _{j})\cdot (U(\theta _{j})-\widehat{U}(\theta _{j})) &=&0
\\
\bar{\gamma}(\theta _{j}) &\geq &0.
\end{eqnarray*}

If $q(\theta _{j})=\widehat{q}(\theta _{j})$ as in (\ref{qbar}), we obtain
the following expressions for the costate variables:%
\begin{eqnarray*}
\gamma (\theta _{j}) &=&-\frac{\lambda JF^{J-1}(\theta _{j})f(\theta
_{j})\cdot \left( 1-\left( 1-\lambda \right) G^{J}(\theta _{j})-\lambda
F^{J}(\theta _{j})\right) }{\left( 1-\lambda \right) JG^{J-1}(\theta
_{j})g(\theta _{j})+\lambda JF^{J}(\theta _{j})f(\theta _{j})} \\
\hat{\gamma}(\theta _{j}) &=&-\frac{\left( 1-\lambda \right) JG^{J-1}(\theta
_{j})g(\theta _{j})\cdot \left( 1-\left( 1-\lambda \right) G^{J}(\theta
_{j})-\lambda F^{J}\right) }{\left( 1-\lambda \right) JG^{J-1}(\theta
_{j})g(\theta _{j})+\lambda JF^{J}(\theta _{j})f(\theta _{j})}.
\end{eqnarray*}%
Differentiating both expressions with respect to $\theta _{j}$ and using the
necessary conditions above, we can solve for the multiplier on the
showrooming constraint $\bar{\gamma}.$ We obtain%
\begin{eqnarray*}
\bar{\gamma}(\theta _{j}) &=&\frac{J\lambda \left( 1-\lambda \right) \left(
1-\left( 1-\lambda \right) G^{J}(\theta _{j})-\lambda F^{J}(\theta
_{j})\right) }{\left( \left( 1-\lambda \right) JG^{J-1}(\theta _{j})g(\theta
_{j})+\lambda JF^{J-1}(\theta _{j})f(\theta _{j})\right) ^{2}}\cdot \\
&&\cdot \left( \frac{\text{d}F^{J-1}(\theta _{j})f(\theta _{j})}{\text{d}%
\theta _{j}}G^{J}(\theta _{j})-\frac{\text{d}G^{J-1}(\theta _{j})g(\theta
_{j})}{\text{d}\theta _{j}}F^{J}(\theta _{j})\right) ,
\end{eqnarray*}%
which is positive if and only if $dF^{J}/dG^{J}$ is increasing in $\theta
_{j}$, i.e., if and only if the $F^{J}$ likelihood-ratio dominates $G^{J}.$%
\bigskip
\end{proof}

\begin{proof}[Proof of Proposition \protect\ref{prop_fullinfo}]
When the platform becomes arbitrarily large $\left( \lambda \rightarrow
1\right) $, rents off-platform vanish and the sponsored seller appropriates
the entire surplus it generates. The sponsored seller's profit under
complete and symmetric information on each value $\theta _{j}$ is then given
by the first-best surplus $\pi _{j}^{\ast }(\theta )=\theta _{j}^{2}/2$. Fix
any matching mechanism, and let $F^{\ast }$ denote the distribution of $%
\theta _{j}$ that are matched to seller $j$. Therefore, the information
design problem of the platform is given by%
\begin{equation*}
\max_{\hat{F}\prec F^{\ast }}\int_{\theta _{L}}^{\theta _{H}}\pi _{j}^{\ast
}\left( \theta _{j}\right) \text{d}\hat{F}\left( \theta _{j}\right) .
\end{equation*}%
Because $\pi _{j}^{\ast }\left( \cdot \right) $ is strictly convex, the
platform-optimal information design sets $\hat{F}=F^{\ast }$, i.e., it
reveals to each consumer their true value for the sponsored seller.
Furthermore, by Proposition \ref{optimalmech}, it is optimal to match
consumers and sellers efficiently (i.e., to further let $F^{\ast }=F^{J}$)
when the platform reveals all the available information.\bigskip
\end{proof}

\begin{proof}[Proof of Proposition \protect\ref{prop_id}]
Fix $\widehat{q}$ and let $\pi (\theta _{j})$ denote the online profit
function. By \cite{dwma19} (Theorem 1), if there exists a distribution $\hat{%
F}\prec F^{J}$ and a convex supporting function $y(\theta _{j})$ such that 
\begin{align*}
y(\theta _{j})& \geq \pi (\theta _{j}) \\
\int_{\theta _{L}}^{\theta _{H}}y(\theta _{j})\text{d}F^{J}(\theta _{j})&
=\int_{\theta _{L}}^{\theta _{H}}y(\theta _{j})\text{d}\hat{F}(\theta _{j})
\\
\text{supp}(\hat{F})& \subseteq \{\theta _{j}\in \left[ \theta _{L},\theta
_{H}\right] :y(\theta _{j})=\pi (\theta _{j})\},
\end{align*}%
then $\hat{F}$ solves our problem for the given $\widehat{q}$.

Because our function $\pi \left( \theta _{j}\right) $ satisfies the \cite%
{dwma19} regularity conditions, computing the supporting function and the
associated distribution is relatively easy: by \cite{dwma19} (Proposition
1), the support of the optimal $\hat{F}$ can be found by solving 
\begin{align*}
& \min_{y}\int_{\theta _{L}}^{\theta _{H}}y(\theta _{j})\text{d}F^{J}(\theta
_{j}) \\
\text{s.t. }& y(\theta _{j})\geq \pi (\theta _{j})\ \forall \theta _{j} \\
& y\text{ convex}.
\end{align*}%
Moreover, the optimal $\hat{F}$ is then supported only on points where $%
y^{\ast }=\pi $. This result allows us to compute the supporting function
independently of the distribution.

To solve the problem, we first reduce it to the choice of one variable,
namely the slope $s$ of the affine function $y$ (when $y\neq \pi $). Call $%
x_{1},x_{2}$ the intersection points of $y$ and $\pi $. Then it holds that 
\begin{align*}
x_{2}^{2}/2-(x_{2}-\mu )\widehat{q}& =\mu ^{2}/2+s(x_{2}-\mu ) \\
\mu ^{2}/2-s(\mu -x_{1})& =x_{1}^{2}/2.
\end{align*}%
Therefore, solving we have 
\begin{equation*}
x_{1}=2s-\mu \text{ and }x_{2}=2s-\mu +2\hat{q},
\end{equation*}%
and we can write the objective as 
\begin{align*}
\min_{s}& \bigg[\int_{0}^{x_{1}(s)}(\theta _{j}^{2}/2)\text{d}F^{J}(\theta
_{j})+\int_{x_{1}(s)}^{x_{2}(s)}\left( x_{1}^{2}(s)/2+s(\theta
_{j}-x_{1}(s))\right) \text{d}F^{J}(\theta _{j}) \\
& +\int_{x_{2}(s)}^{1}\left( \theta _{j}^{2}/2-\widehat{q}(\theta _{j}-\mu
)\right) \text{d}F^{J}(\theta _{j})\bigg].
\end{align*}%
Solving the first-order condition for $s$ in the problem above yields 
\begin{equation*}
\int_{x_{1}(s)}^{x_{2}(s)}(\theta _{j}-\mu )\text{d}F^{J}(\theta _{j})=0.
\end{equation*}%
Finally, from \cite{dwma19} (Theorem 1), we know the support of $\hat{F}%
^{\ast }$ is $[0,x_{1}]\cup \{\mu \}\cup \lbrack x_{2},1]$. Moreover,
duality ensures that the \textit{optimal} supporting function $y^{\ast
}(\theta _{j})$ yields a mean-preserving contraction of $F^{J}$.

Finally, note that%
\begin{align*}
\Pi ^{\prime }(\widehat{q})& =-\lambda \left( \left( x_{2}^{2}-\mu
^{2}\right) /2-\widehat{q}(x_{2}-\mu )\right) dF^{J}(x_{2}(s))x_{2}^{\prime
}\left( q\right) \\
& -\lambda \int_{x_{2}(s)}^{1}(\theta _{j}-\mu )\text{d}F^{J}(\theta
_{j})+\left( 1-\lambda \right) \left( \mu -\widehat{q}\right)
\end{align*}%
hence%
\begin{eqnarray*}
\widehat{q} &=&\mu -\frac{\lambda }{1-\lambda }\left[ \int_{x_{2}(s)}^{1}(%
\theta _{j}-\mu )\text{d}F^{J}(\theta _{j})+\left( \frac{x_{2}^{2}-\mu ^{2}}{%
2}-\widehat{q}(x_{2}-\mu )\right) JF^{J-1}(x_{2})f\left( x_{2}\right)
x_{2}^{\prime }\left( q\right) \right] \\
\widehat{q} &=&\frac{\mu -\frac{\lambda }{1-\lambda }\left[
\int_{x_{2}(s)}^{1}(\theta _{j}-\mu )\text{d}F^{J}(\theta _{j})+\frac{%
x_{2}^{2}-\mu ^{2}}{2}JF^{J-1}(x_{2})f\left( x_{2}\right) x_{2}^{\prime
}\left( q\right) \right] }{1-\frac{\lambda }{1-\lambda }(x_{2}-\mu
)JF^{J-1}(x_{2})f\left( x_{2}\right) x_{2}^{\prime }\left( q\right) }.
\end{eqnarray*}%
Therefore $\widehat{q}$ is decreasing in $\lambda $ and consequently $x_{2}$
is decreasing and $x_{1}$ increasing in $\lambda $.\bigskip
\end{proof}

\begin{proof}[Proofs of Proposition \protect\ref{compstat}]
These results are obtained by differentiating expression (\ref{q0s}) for the
equilibrium quality off-platform with respect to $\lambda $. In particular,
whenever it is strictly positive, the equilibrium $\widehat{q}_{j}^{\ast
}(\theta _{j})$ is strictly decreasing in $\lambda $. Because the
equilibrium quality provision is equal to the marginal information rent, the
comparative statics of quality $\widehat{q}_{j}^{\ast }$ immediately extend
to the information rent $\widehat{U}_{j}^{\ast }$.\bigskip
\end{proof}

\begin{proof}[Proof of Proposition \protect\ref{nJ}]
These results are obtained by differentiating expression (\ref{q0s}) with
respect to $J.$ Whenever it is strictly positive, the equilibrium $\widehat{q%
}_{j}^{\ast }(\theta _{j})$ is strictly decreasing in $J$ for $J$ large
enough. Because the equilibrium quality provision is equal to the marginal
information rent, the comparative statics of quality $\widehat{q}_{j}^{\ast
} $ immediately extend to the information rent $\widehat{U}_{j}^{\ast }$.
\end{proof}

\newpage

\bibliographystyle{apalike}
\bibliography{ale2,general}

\end{document}